\documentclass[12pt]{article}

\usepackage{graphicx}
\usepackage{caption}
\usepackage{amsfonts}
\usepackage{amssymb}
\usepackage{amsmath}
\usepackage{amsthm}
\usepackage{float}
\usepackage{enumerate}
\usepackage{hyperref}

\newtheorem{theorem}{Theorem}[section]
\newtheorem*{theorem*}{Theorem}
\newtheorem{lemma}[theorem]{Lemma}

\let\realbibitem=\bibitem
\def\bibitem{\par \vspace{-1.2ex}\realbibitem}

\graphicspath{{figures/tessellation/}{figures/turning/}}

\title{The Computational Complexity of Finding Hamiltonian Cycles in Grid Graphs of Semiregular Tessellations}

\author{Kaiying Hou\thanks{Phillips Academy Andover, {\tt khou@andover.edu}}
        \and
        Jayson Lynch\thanks{MIT Computer Science and Artificial Intelligence Laboratory, {\tt  jaysonl@mit.edu}}}

\begin{document}

\maketitle

\begin{abstract}

Finding Hamitonian Cycles in square grid graphs is a well studied and important questions. More recent work has extended these results to triangular and hexagonal grids, as well as further restricted versions\cite{HCPsquare,HCPtrihex,TRVBreduction}. In this paper, we examine a class of more complex grids, as well as looking at the problem with restricted types of paths. We investigate the hardness of Hamiltonian cycle problem in grid graphs of semiregular tessellations. We give NP-hardness reductions for finding Hamiltonian paths in grid graphs based on all eight of the semiregular tessilations. Next, we investigate variations on the problem of finding Hamiltonian Paths in grid graphs when the path is forced to turn at every vertex. We give a polynomial time algorithm for deciding if a square grid graph admits a Hamiltonian cycle which turns at every vertex. We then show deciding if cubic grid graphs, even if the height is restricted to $2$, admit a Hamiltonian cycle is NP-complete.

 \end{abstract}

\section{Introduction}
The Hamiltonian cycle problem (HCP) in grid graphs has been well studied and has led to application in numerous NP-hardness proofs for problems such as the milling problem\cite{arkin1993lawnmower}, Pac-Man\cite{viglietta2014gaming}, finding optimal solutions to a Rubix Cube\cite{demaine2017solving}, and routing in wireless mesh networks\cite{6314243}. The problem has been of interest to computer scientists for many years and recently a number of variations on the problem have been investigated. A 1982 paper proved that the HCP in square grid is NP-complete by reducing from the HCP in planar max degree 3 bipartite graphs\cite{HCPsquare}. More recently, a paper published in 2008 proved that the HCPs in triangular and hexagonal grid are NP-complete by the same reduction\cite{HCPtrihex}. In June of 2017, a new paper proved that HCP in hexagonal thin grid graph is NP-complete by reducing from 6-Regular Tree-Residue Vertex Breaking problem\cite{TRVBreduction}. These papers also show results on grid graphs with restrictions such as thin, polygonal, and solid. With all the interest in the computational complexity of the HCP in grid graphs, it is reasonable to ask whether we can generalize or adapt these results to different types of grids. In addition, we investigate the notion of angle-restricted tours, studied in \cite{fekete1997angle}, in the context of grid graphs. We give both algorithms and hardness proofs for finding Hamiltonian paths with this `always turning' constraint.
 
Although the hardness of the HCP in semiregular grids seems like an abstract question, it has many possible applications. Grids are natural structures that things may be formatted into. For example, the layout of buildings or modular structures used in space may form a network that follow the patterns of semiregular, or more general, grids. If certain locations in such net work need to be visited for maintenance, and one wants an optimal rout, then this is well modeled by the Hamiltonian path problem. Our reductions both give insight into what sorts of regular structures will be difficult to find optimal paths for, as well as ways of potentially transferring other efficient algorithms to these new problems. Finally, the results and techniques in this paper may be useful in proving hardness of other problems by reducing from HCPs in semiregular grids. 
 
The always turning Hamiltonian path problems also has some relation to more concrete questions. First, one can see the always turning constraint as path planning in world with reflections at fixed angles and locations. One may be routing optics to various locations on an optics table. Reflections of 45 degrees in a grid based world are also a common element in puzzles and games. In addition, there has been study of problems which try to minimize the number of turns taken in a covering tour\cite{doi:10.1137/S0097539703434267}. In many ways this can be seen as the opposite, modeling a case where turning is significantly easier than continuing straight.

\paragraph*{Results}
In Section~\ref{sec:semireg} we extend the class of grid graphs studied to those based on semiregular tessellations. There are a total of eight semiregular tessellations\cite{naturalstructure}, which are shown in Figure~\ref{fig:semi}. For all eight semiregular tessellations we show the corresponding Hamiltonian cycle problem in the induced grid graph is NP-complete. We show hardness by reducing from three NP-complete problems: HCP in planar max degree 3 bipartite graphs\cite{HCPsquare}, HCP in hexagonal grids\cite{HCPtrihex}, and Tree-Residue Vertex Breaking problem\cite{TRVBreduction}.

In Section~\ref{sec:turning} we examine the question of Hamiltonian paths which turn at every vertex. We show this problem is  hard in 3D square grid graphs. We also show it is easy in triangular grids with only 60 degree turns but hard in triangular grids when 120 degree turns are allowed. Finally, we examine a problem in square grids where a path must visit every vertex at least once, must turn at every vertex, and cannot reuse edges. We give a linear time algorithm for solving this double turning problem in solid square grid graphs.

\section{Definitions}

\begin{figure*}[ht!]
\label{fig:semireg}
\centering
\includegraphics[width=130mm]{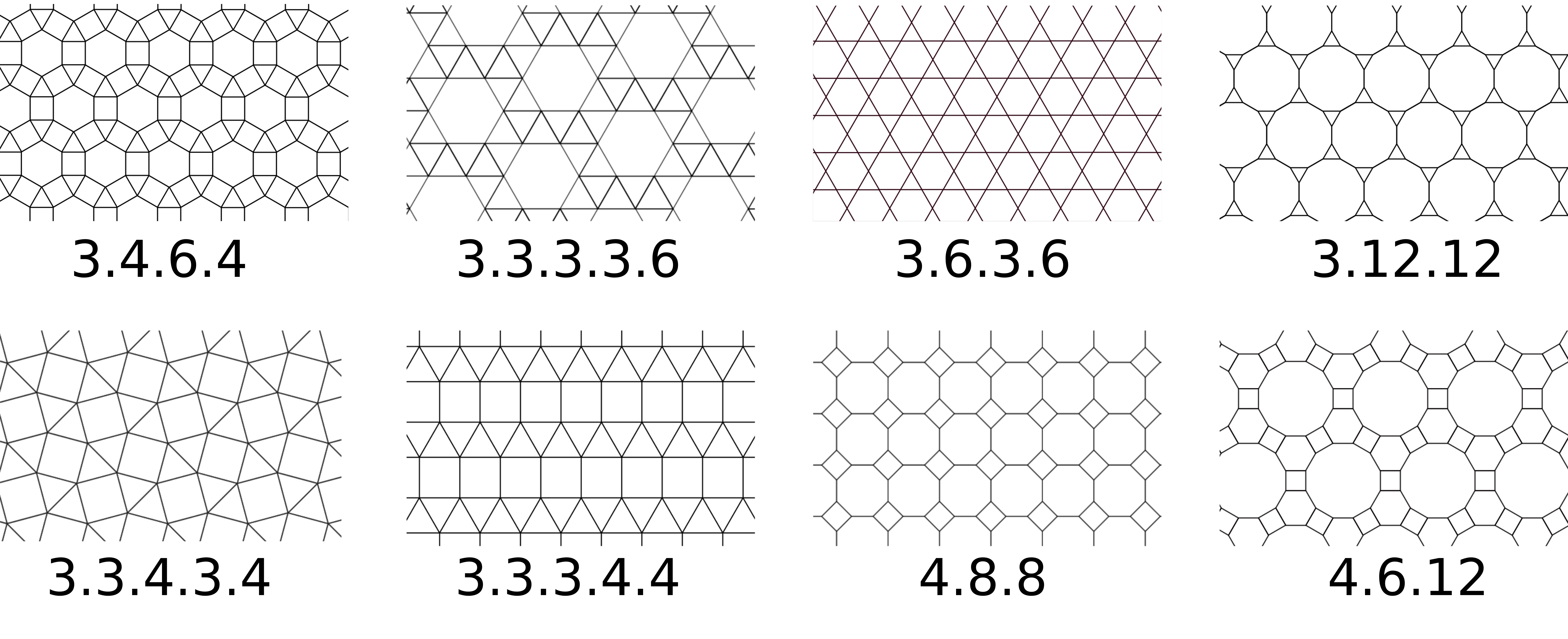}
\caption{The eight semiregular tessellations. It is common to refer to them by the size of the faces while walking around a vertex. }
\label{fig:semi}
\end{figure*}

A \emph{tessellation} is a tiling of a plane with polygons without overlapping. A \emph{semiregular tessellation} is a tessellation which is formed by two or more kinds regular polygons of side length 1 and in which the corners of polygons are identically arranged. Figure~\ref{fig:semi} depicts part of each of the eight semiregular tessellations.

 An infinite lattice of a semiregular tessellation is a lattice formed by taking the vertices of the regular polygons in the tessellation as the points of the lattice. A graph $G$ is \emph{induced} by the point set $S$ if the vertices of $G$ are the points in $S$ and its edges connect vertices that are distance 1 apart. A \emph{grid graph} of a semiregular tessellation, or a \emph{semiregular grid}, is a graph induced by a subset of the infinite lattice formed by that tessellation. Call the infinite graph induced by the full lattice a \emph{full grid}.\\

A \emph{pixel} is the simple cycle bounding a face in a grid graph which contains the same bounding edges and vertices as the corresponding face in the full grid. Thus a pixel can be thought of as a cycle in a graph which bounds precisely one tile in the original tessellation. We may use pixel interchangeably to refer to the bounding cycle, the face bound by the cycle, or the set of vertices around that face.
A \emph{solid grid graph} is one in which every bounded face is a pixel. 
 
 A \emph{Hamiltonian cycle} is a cycle that passes through each vertex of a graph exactly once. The \emph{Hamiltonian cycle problem}, sometimes abbreviated as HCP, asks that given a graph, whether or not that graph admits a Hamiltonian cycle. The HCP in a semiregular tessellation asks, given a grid graph of that tessellation, whether it admits a Hamiltonian cycle.

\section{Finding Hamiltonian Paths in Semi-Regular Tessellations is NP-Complete}
\label{sec:semireg}
This section shows the NP-completeness of HCPs in all eight semiregular tessellations. There are three NP-complete problems that we reduce from: the HCP in hexagonal grid, the HCP in planar max degree 3 bipartite graphs, and the tree-residue vertex-breaking (TRVB) problem. This section is divided into three subsections, each of which introduces one of the three NP-complete problems and includes the hardness proofs that reduce from that problem. \\

\subsection{HCPs that Reduce from the HCP in Hexagonal Grid}
\label{sec:HexGrid}
This section proves that the HCPs in the 3.4.6.4 tessellation, 3.3.3.3.6 tessellation, 3.6.3.6 tessellation, and 3.12.12 tessellation are NP-complete by reducing from the HCP in hexagonal grids, which is proven to be NP-complete\cite{HCPtrihex}. The reduction works in the following way: for any given Hexagonal grid graph $G'$, which is sometimes referred to as the original graph, we can construct a simulated grid graph $G$ of the target tessellation that has a Hamiltonian cycle if and only if $G'$ has a Hamiltonian cycle. The grid graph $G$ is constructed by using gadgets to represent vertices and edges of the original graph $G'$.

\subsubsection{3.4.6.4 Tessellation}
\label{sec:3464}
\begin{figure}[H]
\centering
\includegraphics[width=70mm]{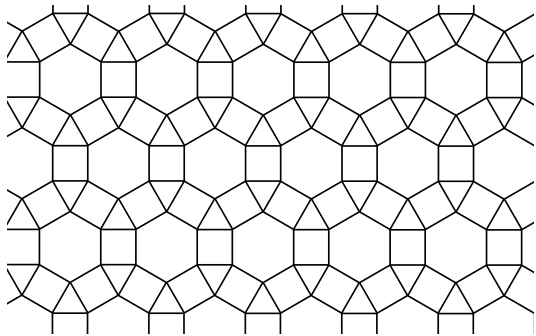}\\
\caption{3.4.6.4 Tessellation}
\end{figure}
\begin{theorem}
\label{thm:3464}
The HCP in grid graphs of the 3.4.6.4 tessellation is NP-complete.\\
\end{theorem}
\begin{proof}
We will reduce from the HCP in hexagonal grids. Given a hexagonal grid graph $G'$, we will construct a grid graph $G$ of the 3.4.6.4. tessellation in this way: for every edge in $G'$ we add the edge gadget shown in Figure~\ref{fig:3464edge} to $G$ and for every vertex in $G'$ we add the vertex gadgets shown in Figure~\ref{fig:3464vertex} to $G$. Since the 3.4.6.4 tesselation has scaled versions of the translational symmetries of the hexagonal grid, picking an embedding for our construction is easy. An example can be seen in Figure~\ref{fig:3464example}. Since the hexagonal grid $G'$ is bipartite, we can design different vertex gadgets for the even and odd vertices.\\

\begin{figure}[H]
    \centering
    \includegraphics[width=.5\textwidth]{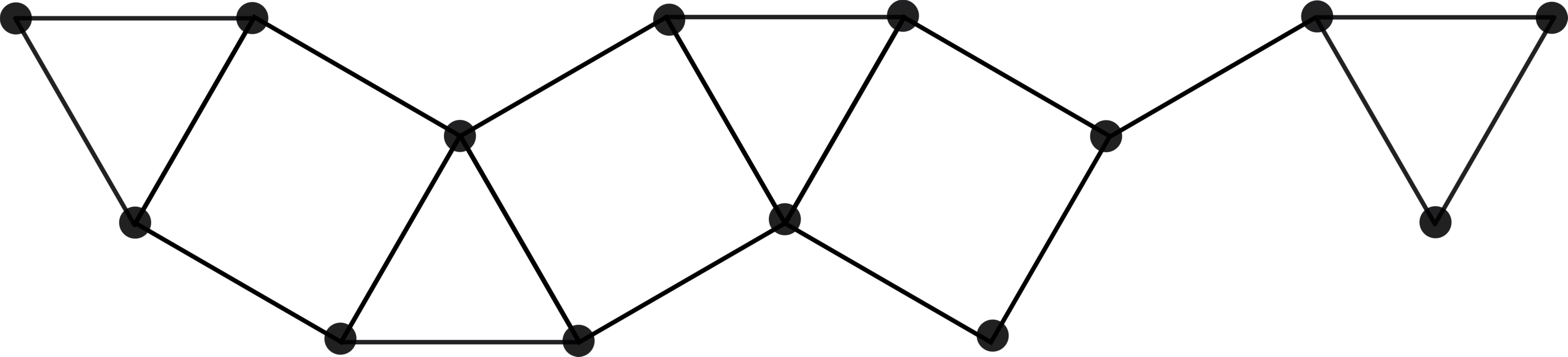}
    \caption{Edge gadget}
    \label{fig:3464edge}
\end{figure}
\begin{figure}[h]
    \centering
    \includegraphics[width=.6\textwidth]{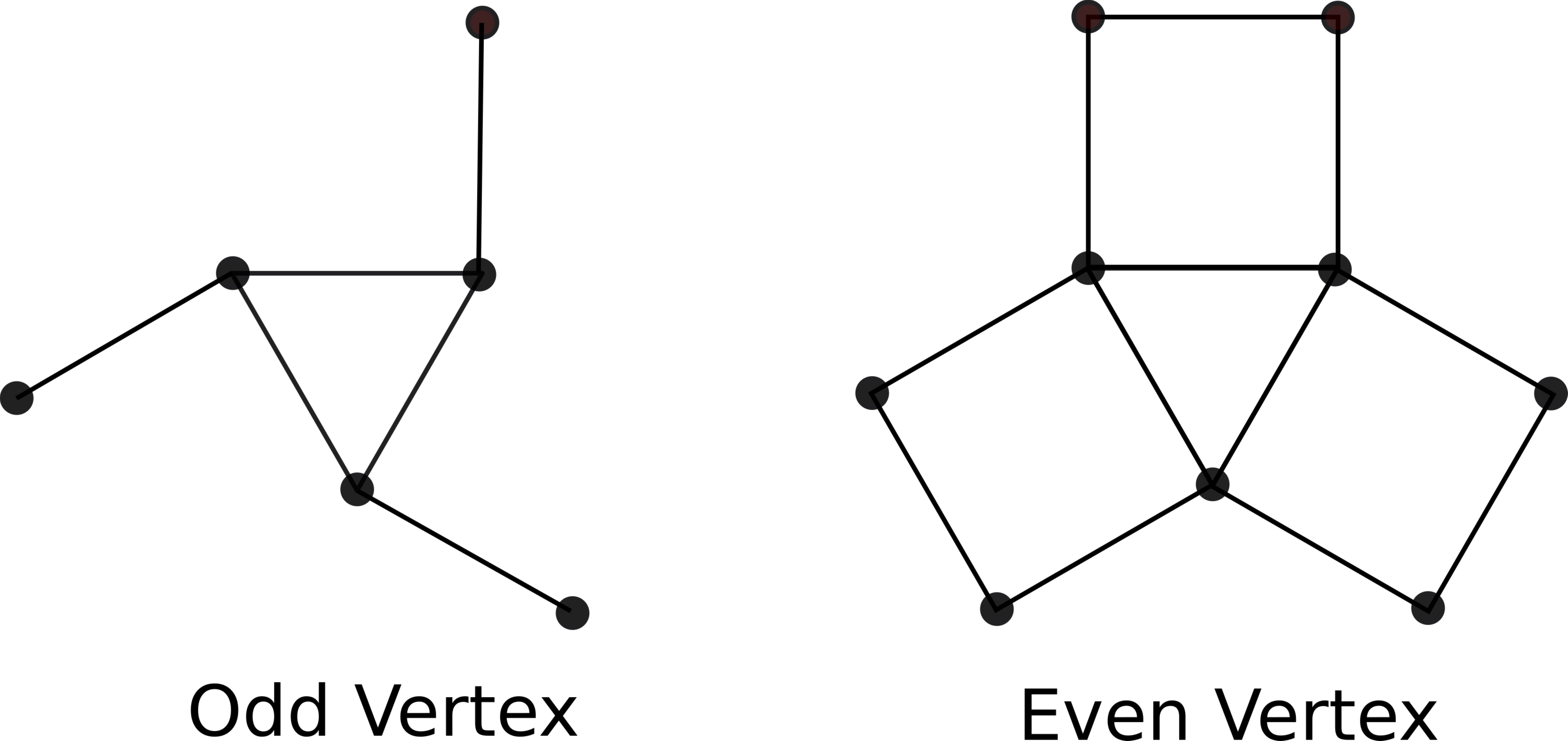}
    \caption{Vertex gadgets}
    \label{fig:3464vertex}
\end{figure}

\begin{figure}[H]
\centering
\includegraphics[width=75mm]{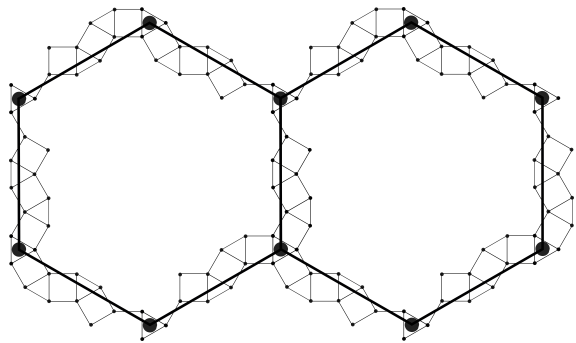}\\
\caption{A simulated graph}
\label{fig:3464example}
\end{figure}

Now, we will show that the original graph $G'$ has a Hamiltonian cycle $C'$ if and only if the simulated graph $G$ has a Hamiltonian cycle $C$. If the $G'$ has a Hamiltonian cycle $C'$, for any taken edge in it, we go through the corresponded edge gadget in $G$ with the cross path in Figure~\ref{fig:3464edgetake}; for any untaken edge, we go through the corresponded edge gadget with the return path. Because the simulated vertices in $G$ are triangles ($K_3$), there is always a path to take the simulated vertex by entering from one point and leaving at the other. Therefore, if there is a Hamiltonian cycle $C'$ in the original graph $G'$, then there is a Hamiltonian cycle $C$ in the simulated graph $G$.\\

\begin{figure}[H]
\centering
\includegraphics[width=80mm]{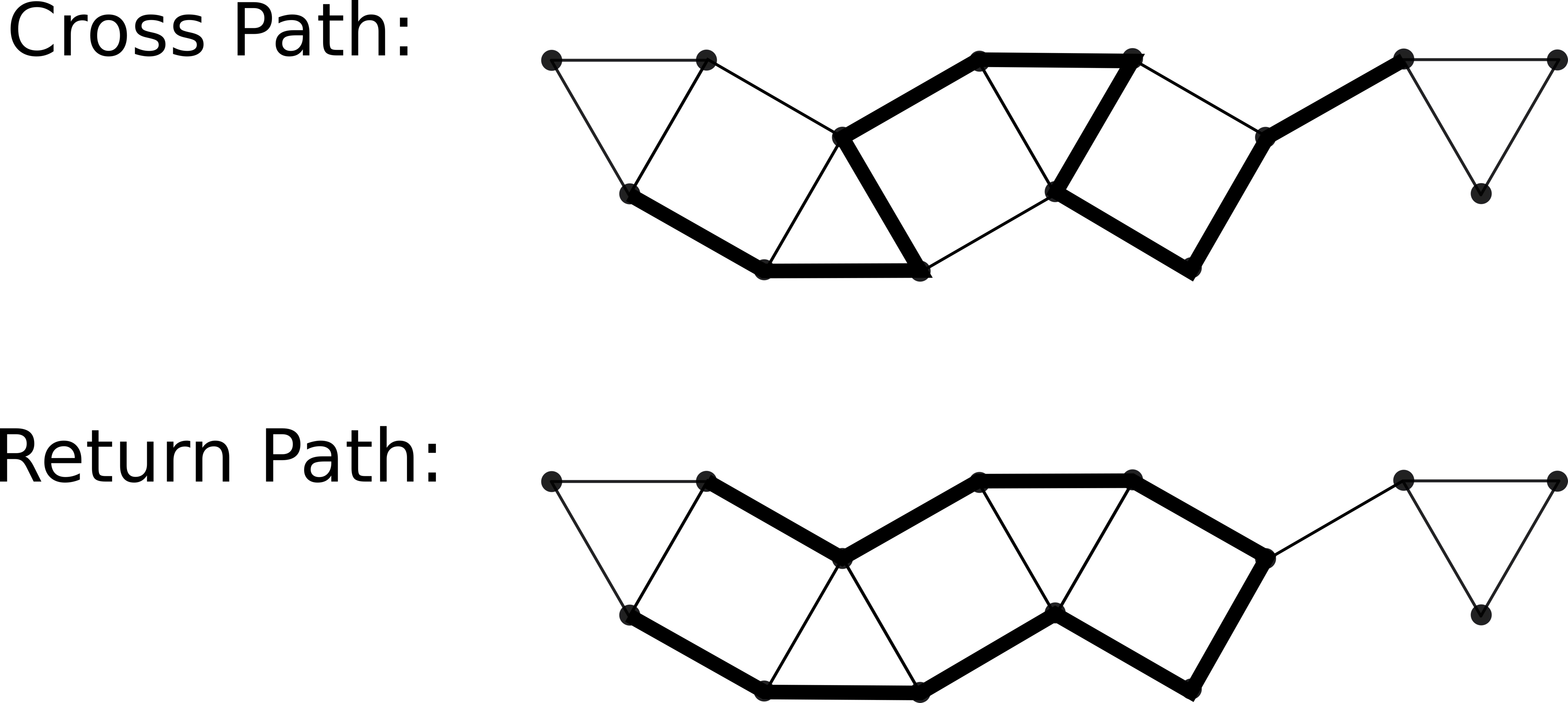}\\
\caption{Two kinds of paths in the 3.4.6.4 edge gadget}
\label{fig:3464edgetake}
\end{figure}

The essential difference between the cross path and the return path is that a cross path starts and finishes at different ends of an edge while the return path starts and finishes in the same end. Note that the return and cross paths are the only two paths which go through the edge gadget and visit all of its vertices. The odd vertex gadget is connected to the edge gadget through a single edge connection which prevents the return path from entering the odd vertex gadget. If a Hamiltonian cycle $C$ exists in the simulated graph $G$, it is not hard to see that each odd vertex gadget in $G$ must be connected to two cross paths and the even vertex gadgets can either be connected to two cross paths or two cross paths and a return path. Then, we can find a cycle $C'$ in the original graph $G'$ by making each edge gadget with a cross path in $C$ a taken edge in $C'$. Thus, if there is a cycle $C$ in the simulated graph $G$, there is a cycle $C'$ in the original graph $G'$. This way, we showed the original graph $G'$ has a Hamiltonian cycle $C'$ if and only if the simulated graph $G$ has a Hamiltonian cycle $C$. \\
\begin{figure}[H]
\centering
\includegraphics[width=80mm]{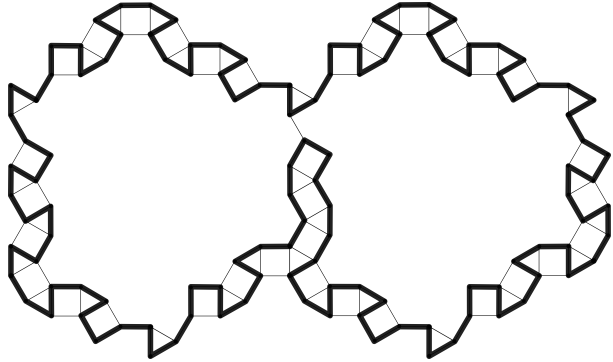}\\
\caption{Example of a reduction for the 3.4.6.4 tessellation}
\label{fig:3464pathexample}
\end{figure}
\end{proof}

\subsubsection{3.3.3.3.6 Tessellation}
\begin{figure}[h]
\centering
\includegraphics[width=70mm]{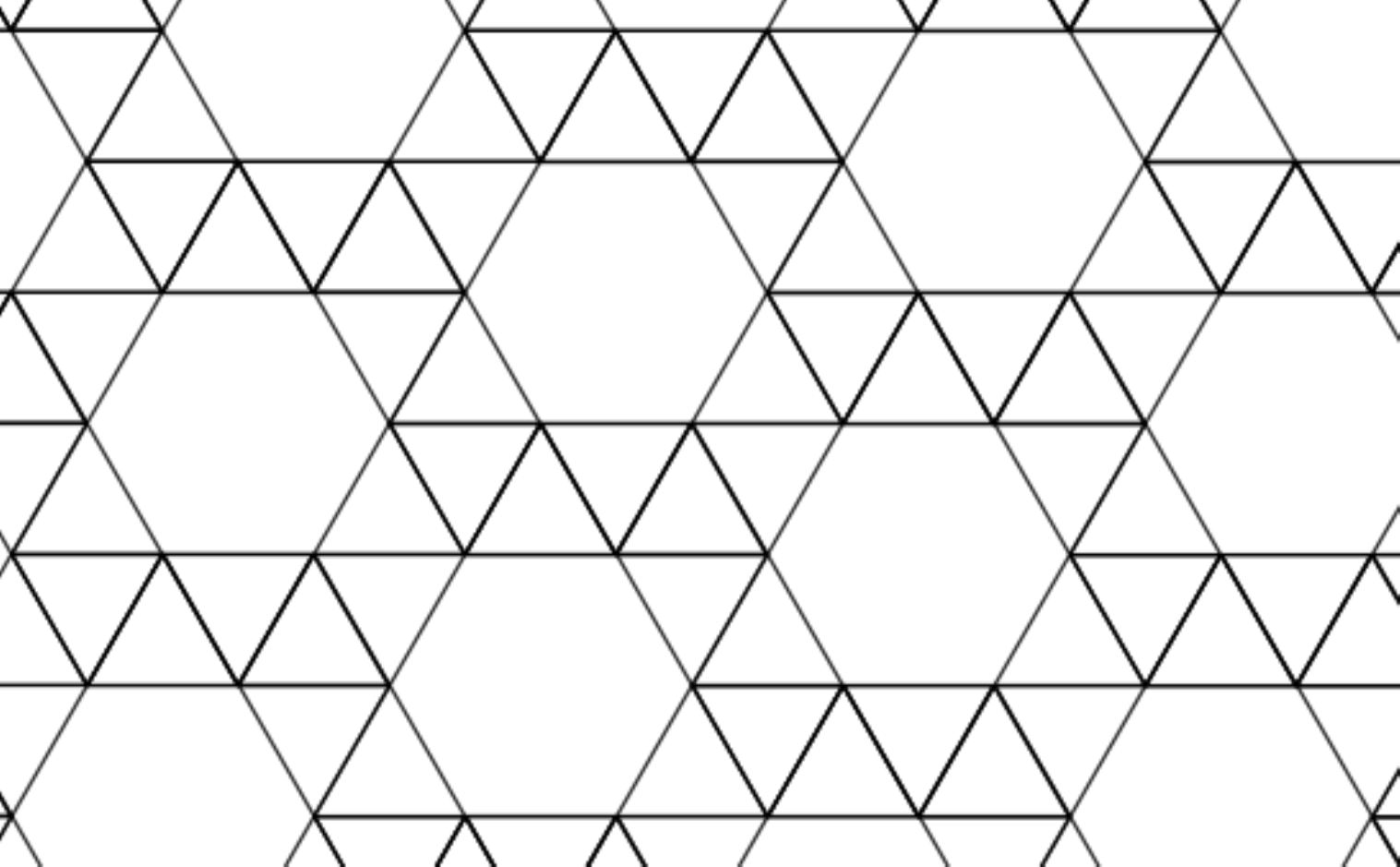}\\
\caption{3.3.3.3.6 Tessellation}
\end{figure}
\begin{theorem}
The HCP in the grid graphs of the 3.3.3.3.6 tessellation is NP-complete.\\
\end{theorem}
\begin{proof}
Similar to the 3.4.6.4 tessellation in Section~\ref{sec:3464}, the NP-completeness of the HCP tessellation can also be proven by reducing from the HCP in hexagonal grid. We use the gadgets shown in Figure~\ref{fig:33336gadgets} to simulate the vertices and edges of the hexagonal grid. Now, we can construct a simulated graph $G$ for any hexagonal grid $G'$. For example, the graph formed by two hexagons can be simulated by the grid in Figure~\ref{fig:33336example}.\\
\begin{figure}[h]
\centering
\includegraphics[width=100mm]{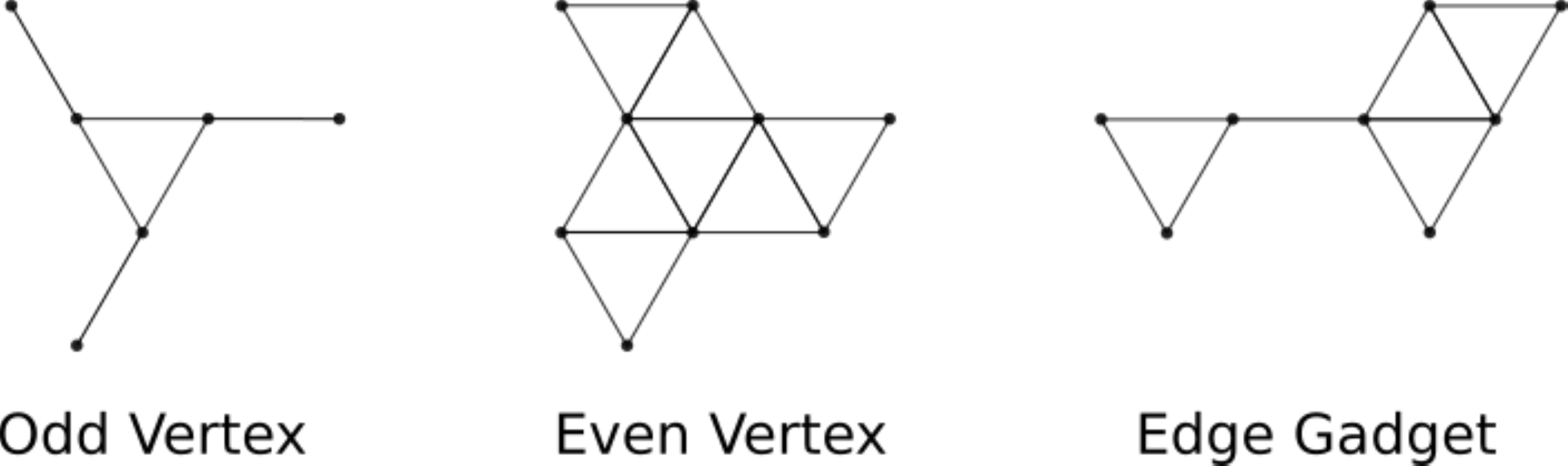}\\
\caption{Gadgets for the 3.3.3.3.4 reduction}
\label{fig:33336gadgets}
\end{figure}
\begin{figure}[H]
\centering
\includegraphics[width=60mm]{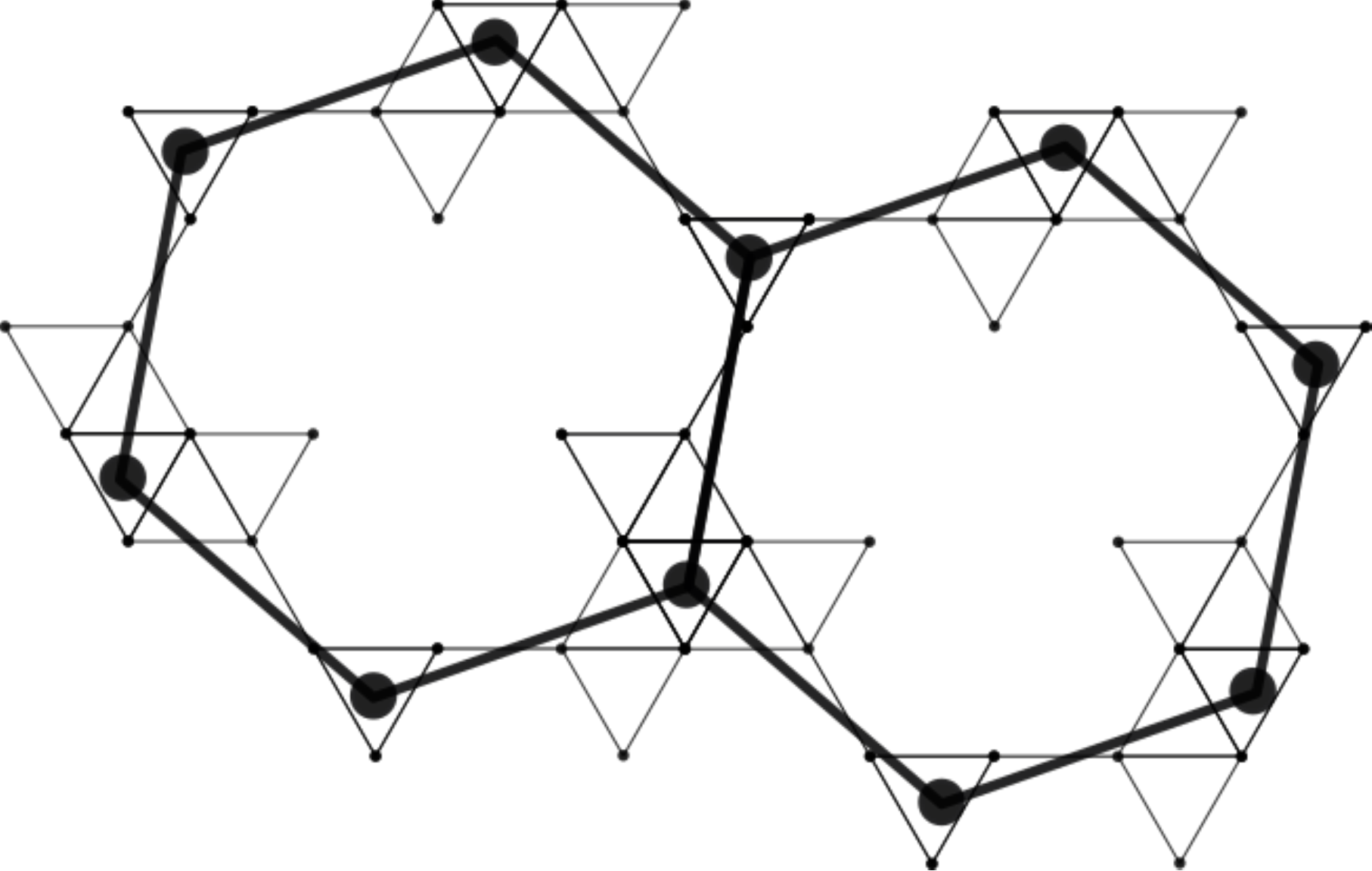}\\
\caption{An example of a simulated graph in the 3.3.3.3.4 tessellation}
\label{fig:33336example}
\end{figure}
Similar to the gadgets used in Section~\ref{sec:3464}, there are two kinds of traversals for the edge gadget: a cross path that goes from one end to the other end and the return path that begins and finishes on the same end. The following reasoning on why $G$ has a Hamiltonian cycle if and only $G'$ has a Hamiltonian cycle is identical to that of the previous section. If a hexagonal grid $G'$ has a Hamiltonian cycle, we can create a Hamiltonian cycle in $G$ by going through the edge gadgets of the taken edges with cross paths and and the edge gadgets of the untaken edges with return paths. If there is a Hamiltonian cycle in $G$, each vertex gadget of $G$ must be connected to exactly two cross paths, indicating that there exists a Hamiltonian cycle in $G$.\\
\end{proof}

\subsubsection{3.6.3.6 Tessellation}
\begin{figure}[h]
\centering
\includegraphics[width=70mm]{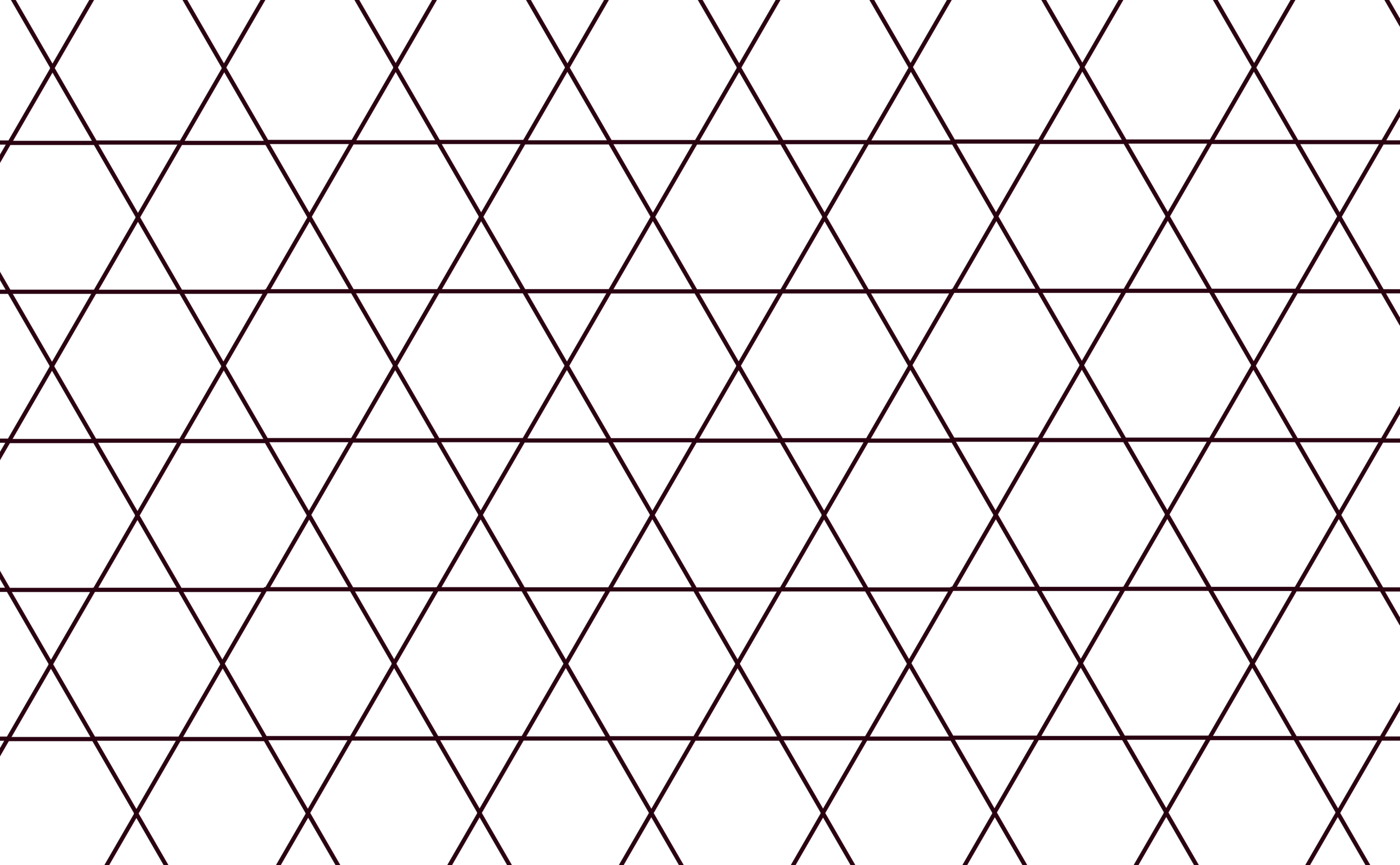}
\end{figure}
\begin{theorem}
The HCP in the grid graphs of the 3.6.3.6 tessellation is NP-complete.\\
\end{theorem}
\begin{proof}
We prove that the HCP in this tessellation is NP-complete by reducing from HCP in hexagonal grid. Using the following vertex gadgets and edge gadget, shown in Figures~\ref{fig:6363vertex} and Figure~\ref{fig:6363edge}, for any hexagonal grid $G'$ we can construct a simulated graph $G$ in the tessellation.\\
\begin{figure}[h]
\centering
\includegraphics[width=80mm]{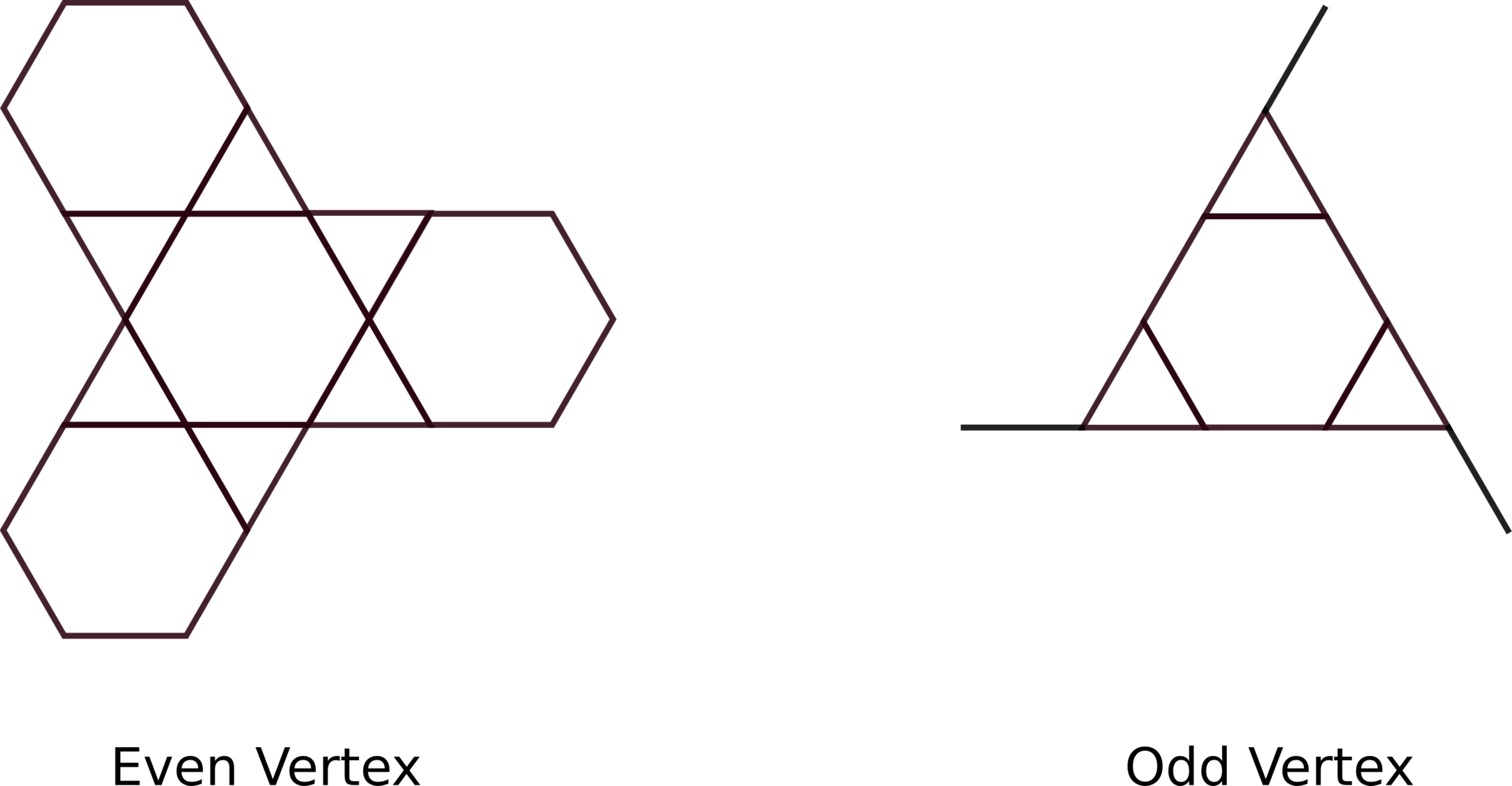}
\caption{Vertex gadgets for 6.3.6.3}
\label{fig:6363vertex}
\end{figure}
\begin{figure}[h]
\centering
\includegraphics[width=70mm]{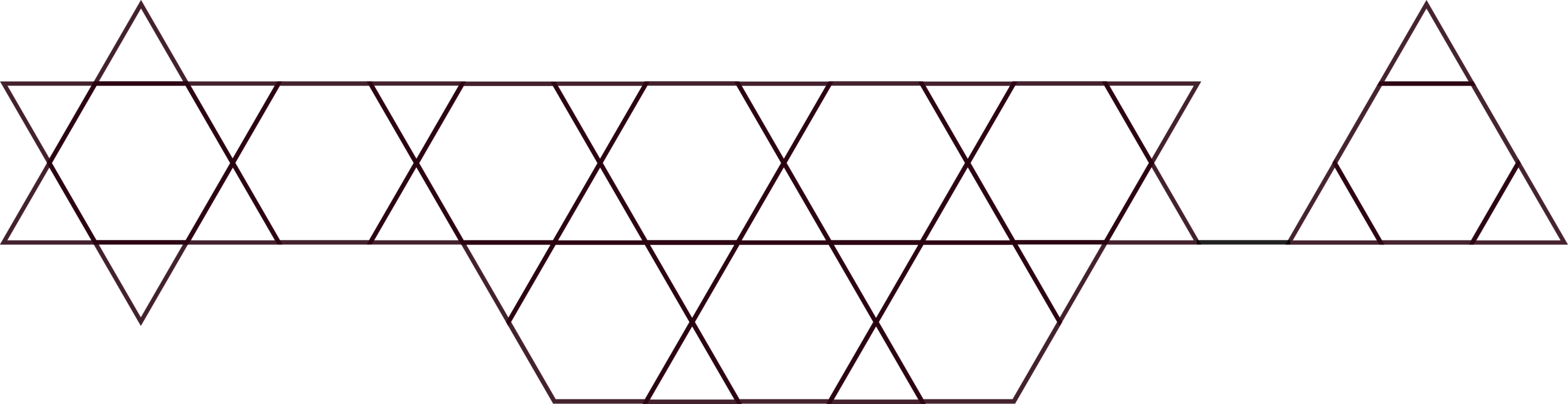}
\caption{Edge gadgets for 6.3.6.3}
\label{fig:6363edge}
\end{figure}
 Each edge gadget has two kinds of traversals: return paths and cross paths. Return paths begin and end on the same end of the edge while cross paths start and finish on different ends. With some inspection, it is clear that return paths and cross paths are the only two kinds of traversals allowed in the edge gadget. Figure~\ref{fig:6363return} shows a possible return path. Different from those of previous tessellations, the edge gadget here has two kinds of cross paths as shown in Figure~\ref{fig:6363cross}. Although the two kinds of cross paths start out the same from the odd vertex gadget on the right, they finish in the even vertex on the left differently. The way a cross path connects to an even vertex gadgets dictates which direction it can go next. The upper cross path must turns clockwise when going through the even vertex, allowing it to connect to an upper cross path while the lower one must turn counter-clockwise, allowing it to connect to a lower cross path. By choosing the correct kind of cross paths, any pair of the three edges of the even vertex gadget can be taken by compatible cross paths. By inspection, we can easily see that odd vertex gadget can connect to any pair of the three edges in two cross paths as well. \\
 \begin{figure}[h]
\centering
\includegraphics[width=70mm]{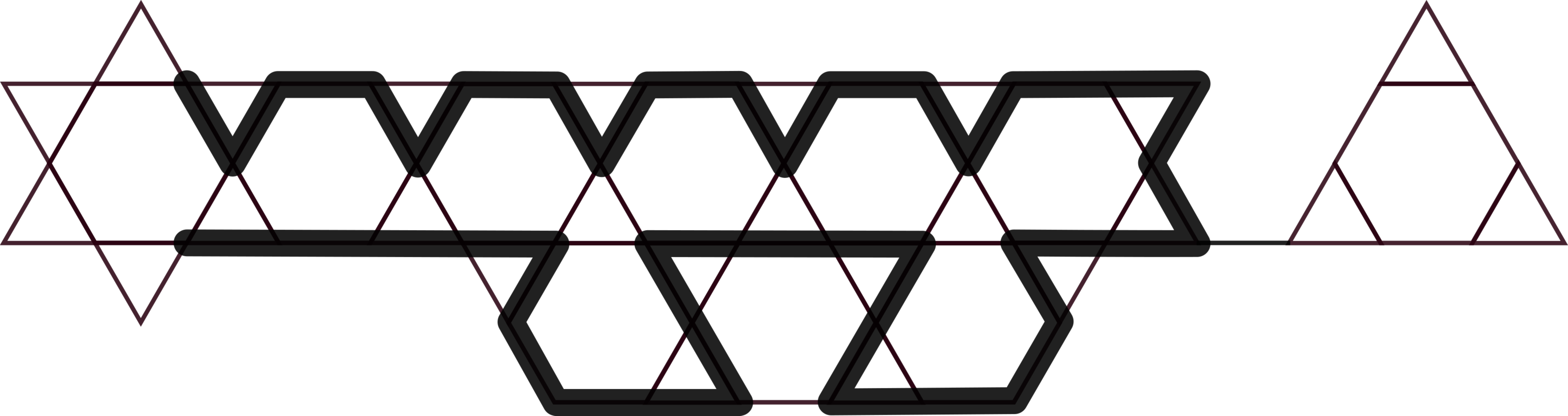}
\caption{Return path for 6.3.6.3 edge gadget}
\label{fig:6363return}
\end{figure}
\begin{figure}[h]
\centering
\includegraphics[width=70mm]{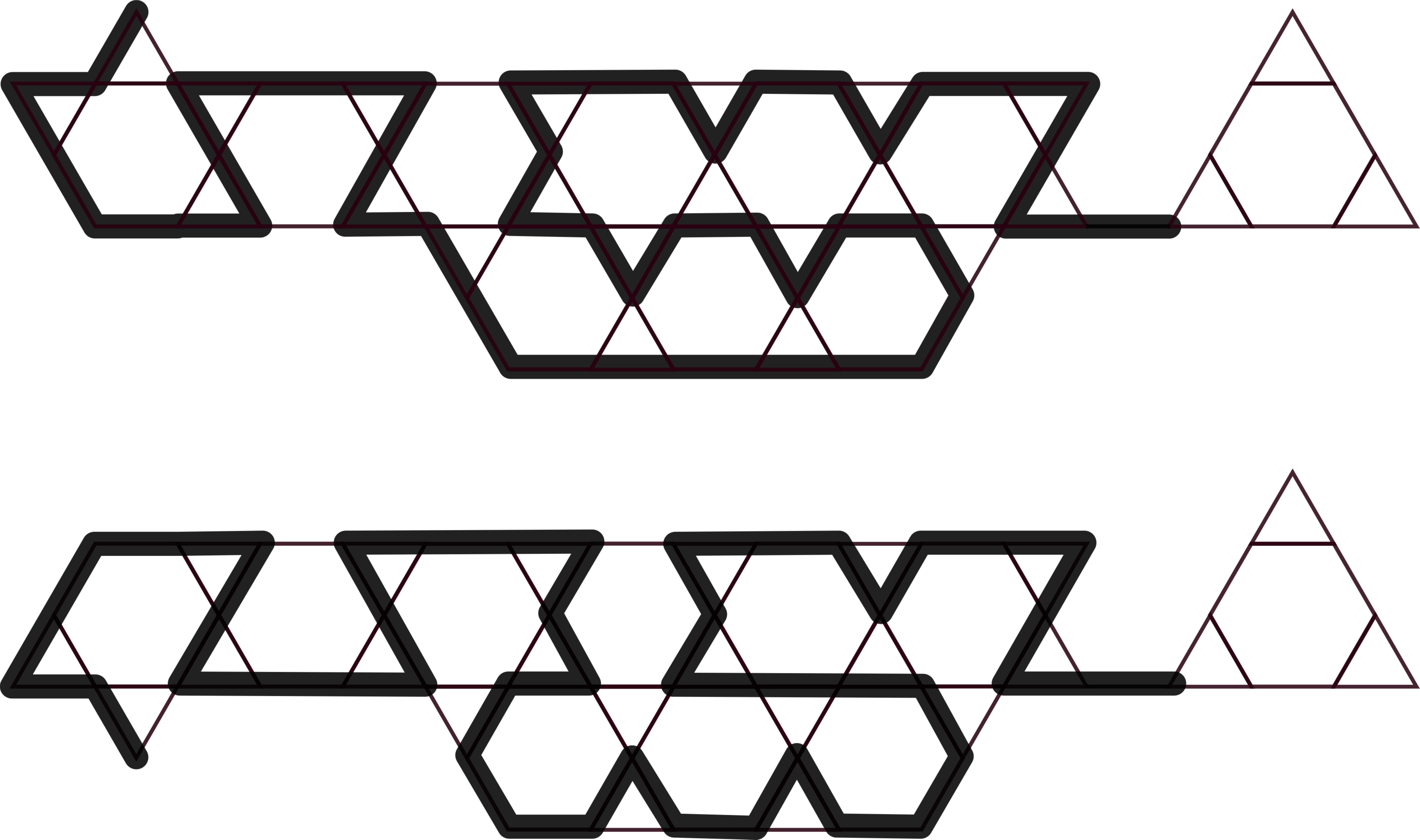}
\caption{Two Kinds of Cross Paths for 6.3.6.3 edge gadget}
\label{fig:6363cross}
\end{figure}
 Now, we will show that the simulated graph $G$ has a Hamiltonian cycle if and only if the original graph $G'$ has a cycle. If the original hexagonal grid $G'$ has a cycle $C'$, then we go through the edges gadgets representing taken edges in $C'$ with a cross path and those representing untaken edges with a return path. Note that we need to use the correct kind of cross paths so that the choice matches the turning at the vertex. If so, then there is a also a Hamiltonian cycle in $G$. If the simulated graph $G$ has a Hamiltonian cycle, each vertex gadget must be connected to exactly two cross paths, which indicate that there is a Hamiltonian cycle in $G'$.\\
\end{proof}

\subsubsection{3.12.12 Tessellation}
\begin{figure}
\centering
\includegraphics[width=70mm]{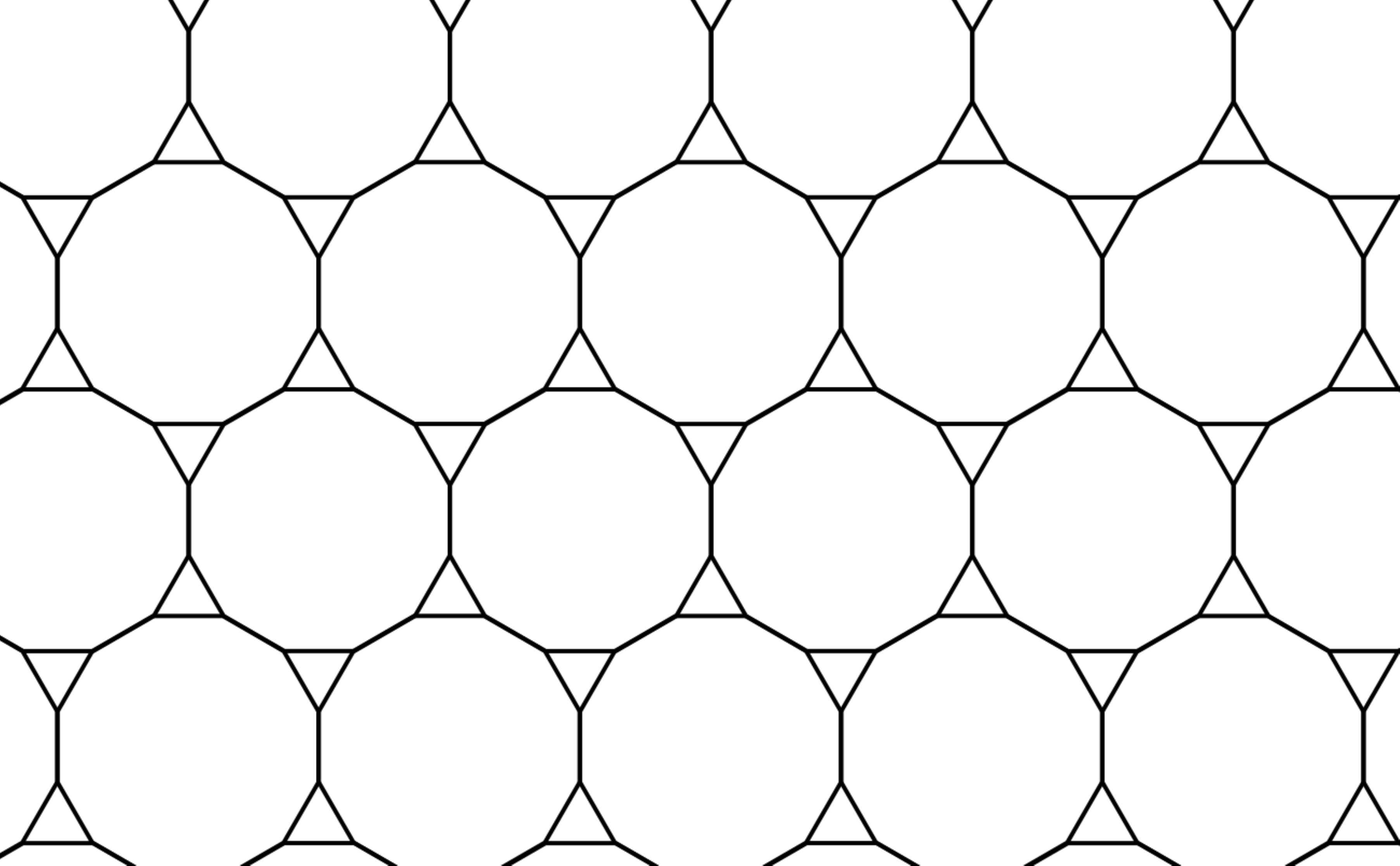}
\caption{3.12.12 Tessellation}
\end{figure}
\begin{theorem}
The HCP in the grid graphs of the 3.12.12 tessellation is NP-complete.\\
\end{theorem}
\begin{proof}
This tessellation is composed of dodecagons and triangles. For a hexagonal grid $G'$, we construct a simulated graph $G$ in the tessellation by using the triangles as vertex of $G'$ and the edges in between triangles as the edges of $G'$. If a Hamiltonian cycle exists in $G$, each triangle must be connected to two paths that form a $120^\text{o}$ angle. Then, there must also be a Hamiltonian cycle in the hexagonal grid $G'$. If there is a Hamiltonian path in the hexagonal grid $G'$, then there exist one in $G$.\\
\end{proof}

\subsection{HCPs that Reduce from the HCP in Planar Max Degree 3 Bipartite Graph}
This section proves that the HCPs in 3.3.4.3.4 tessellation and 3.3.3.4.4. tessellation are NP-complete by reducing from the NP-complete HCP in planar max degree 3 bipartite graph\cite{HCPsquare}. For any given Planar Max Degree 3 Bipartite Graph $G'$, which is sometimes referred to as the original graph, we can construct a grid graph $G$ of the tessellation that has a Hamiltonian cycle if and only if $G'$ has a Hamiltonian cycle. When constructing $G$, we again use gadgets to simulate the edges and vertices of the original graph $G$.

\subsubsection{3.3.4.3.4 Tessellation}
\begin{figure}[H]
\centering
\includegraphics[width=70mm]{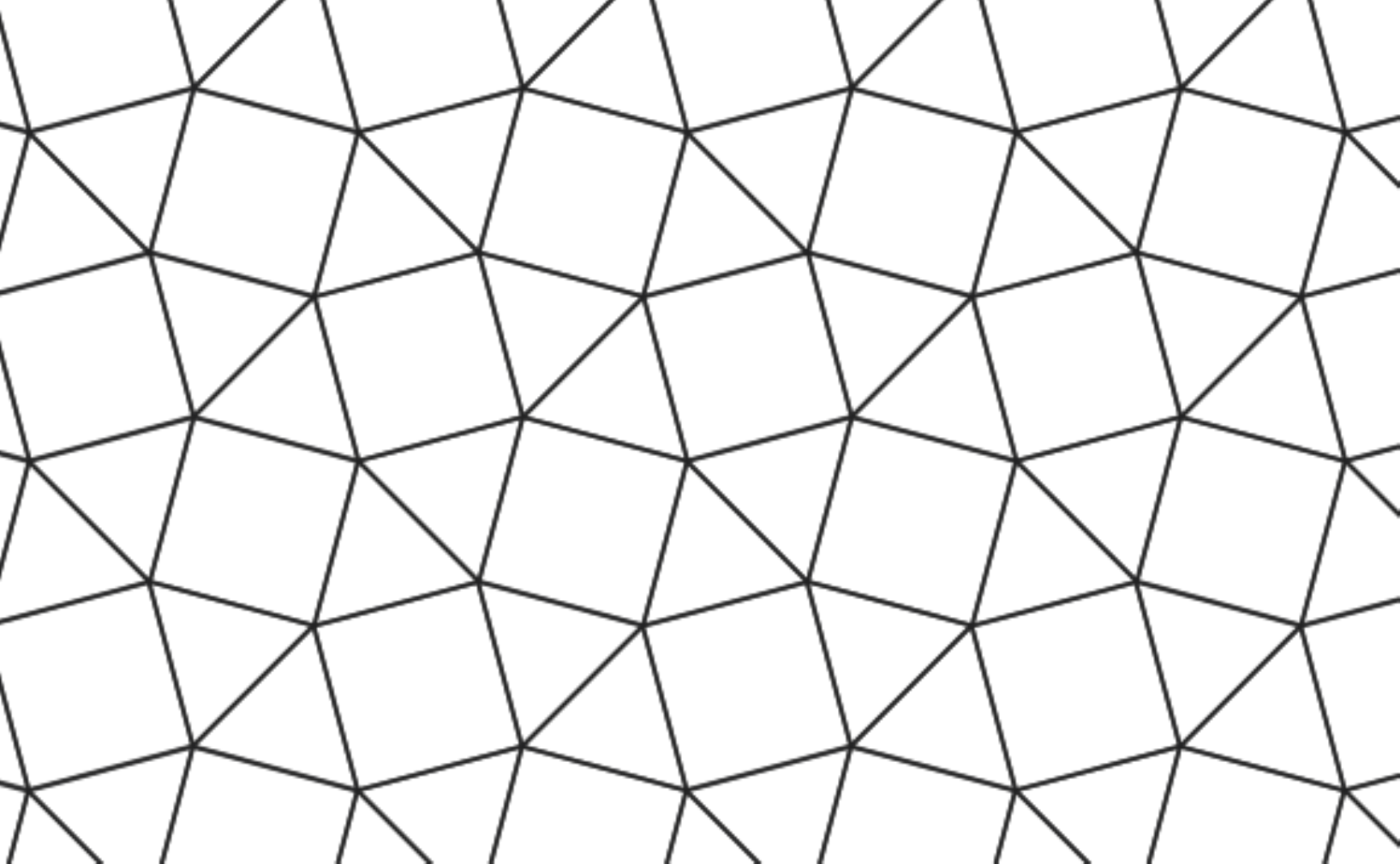}\\
\end{figure}
\begin{theorem}
\label{thm:33434}
The HCP in the grid graphs of the 3.3.4.3.4 tessellation is NP-complete.\\
\end{theorem}
\begin{proof}
We will reduce from HCP in planar max-degree-3 bipartite graphs. First observe that this tessellation can be viewed as a square grid with some extra diagonals. We directly use the gadgets of the square grid proof in the 1982 paper for constructing $G$\cite{HCPsquare}. The edge, even vertex and odd vertex gadgets are shown below. Note that these gadgets are identical to the square grid gadgets except they have some extra edges. In creating the simulated graph $G$ based on a planar max degree 3 bipartite graph $G'$, we go through the same process as that in the square grid reduction: first create a parity-preserving embedding of the max degree 3 bipartite graph; then replace the edges and vertices of the embedding with respective gadgets\cite{HCPsquare}.\\
\begin{figure}[H]
\centering
\includegraphics[width=.6\textwidth]{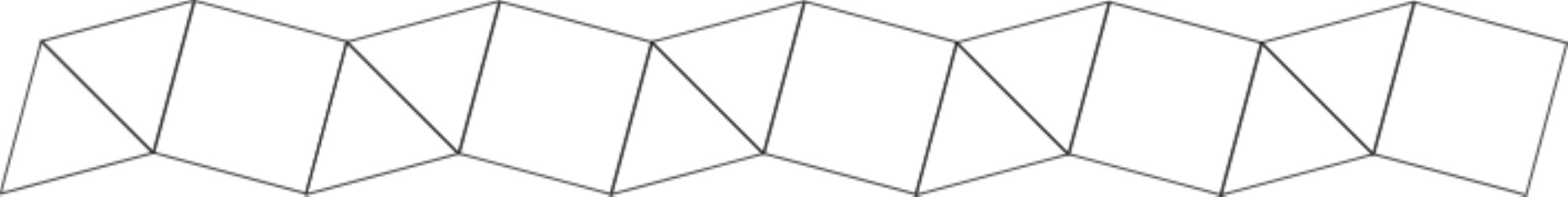}\\
\caption{Edge gadget for 3.3.4.3.4}
\label{fig:33434edge}
\end{figure}
\begin{figure}[h]
\centering
\includegraphics[width=\textwidth]{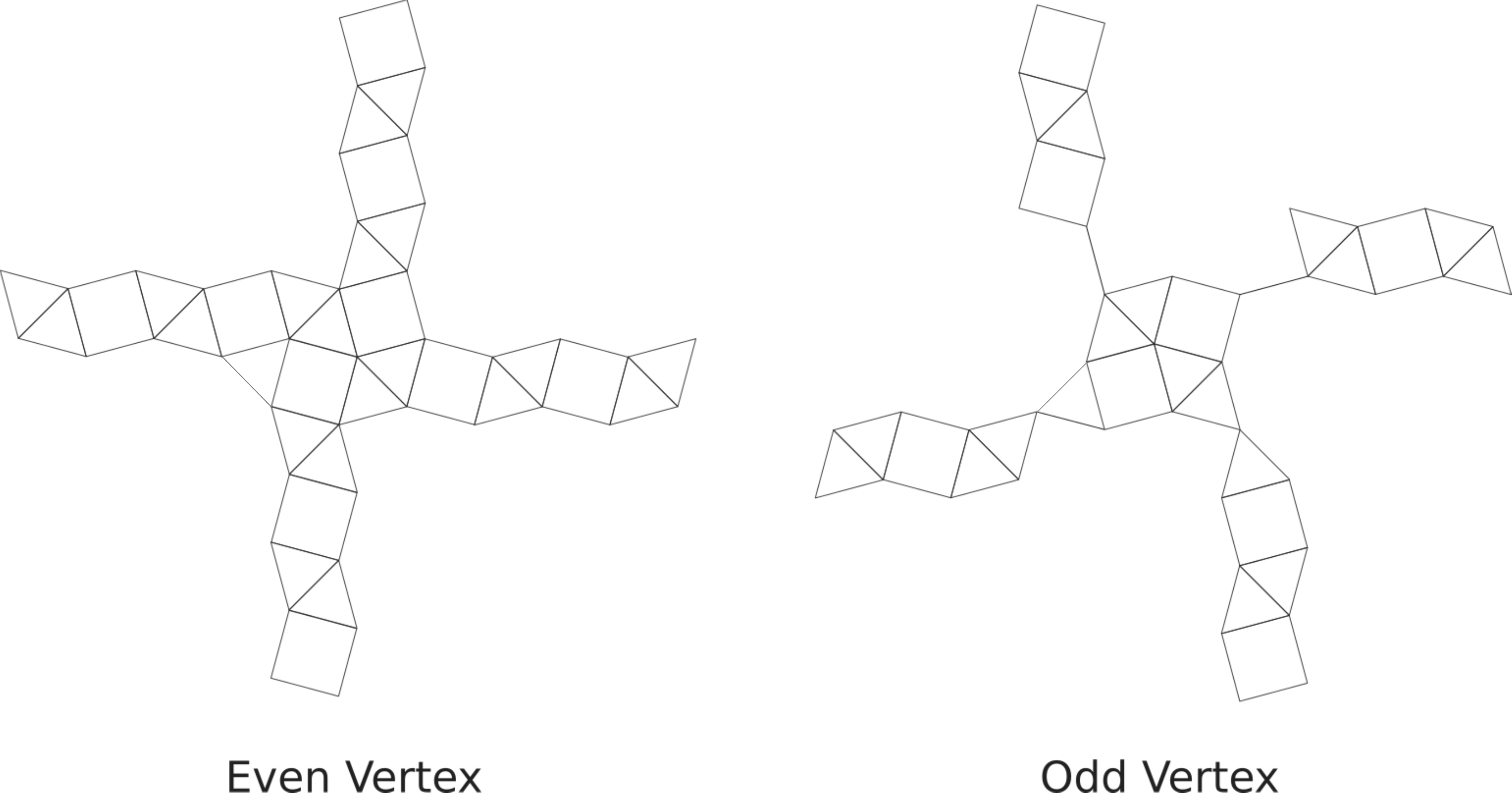}\\
\caption{Vertex gadgets for 3.3.4.3.4}
\label{fig:33434vertex}
\end{figure}
It is not hard to see that there are only two kinds of traversals for the edge gadget: cross paths and a return paths. Although there is more than one kind cross path due to the extra edges, they have the essential characteristic of starting from one end of the gadgets and finishing at the other end (unlike the return path that begins and finishes at the same end). Another difference from the square grid reduction is that the odd vertex gadgets connect to the bottom edge gadget through a single point rather than a single edge as the other edge gadgets. This single point connection also prevents a return path from entering the odd vertex gadget. The single edge and single point connections are called pin connections. It is clear from the pin connections, that we can only have the path enter and exit odd vertices once. Since the graph is bipartite, this forces the other two edges to be return paths, ensuring out simulated path can only enter and exit each vertex once.\\

\end{proof}

\begin{theorem}
The HCP in the grid graphs of the 3.3.3.4.4 tessellation is NP-complete.\\
\end{theorem}
\begin{proof}
Similar to the 3.3.4.3.4 tessellation, this tessellation can also be considered as a square grid with extra diagonals. Because its resemblance to square grid, we again use the square grid gadgets. However, if we use the same reduction as in \cite{HCPsquare}, an extra diagonal may disable a pin connection, being an extra edge that connects the odd vertex gadget with the edge gadget. Then, a return path can enter into the odd vertices through this extra edge, causing the former pin connection to no longer function.
The connection to the upper edge gadget in an odd vertex gadget shown in Figure~\ref{fig:33344parrallel} is an example of a disabled pin connection. Thus, we will need to modify the reduction. \\
\begin{figure}[H]
\centering
\includegraphics[width=60mm]{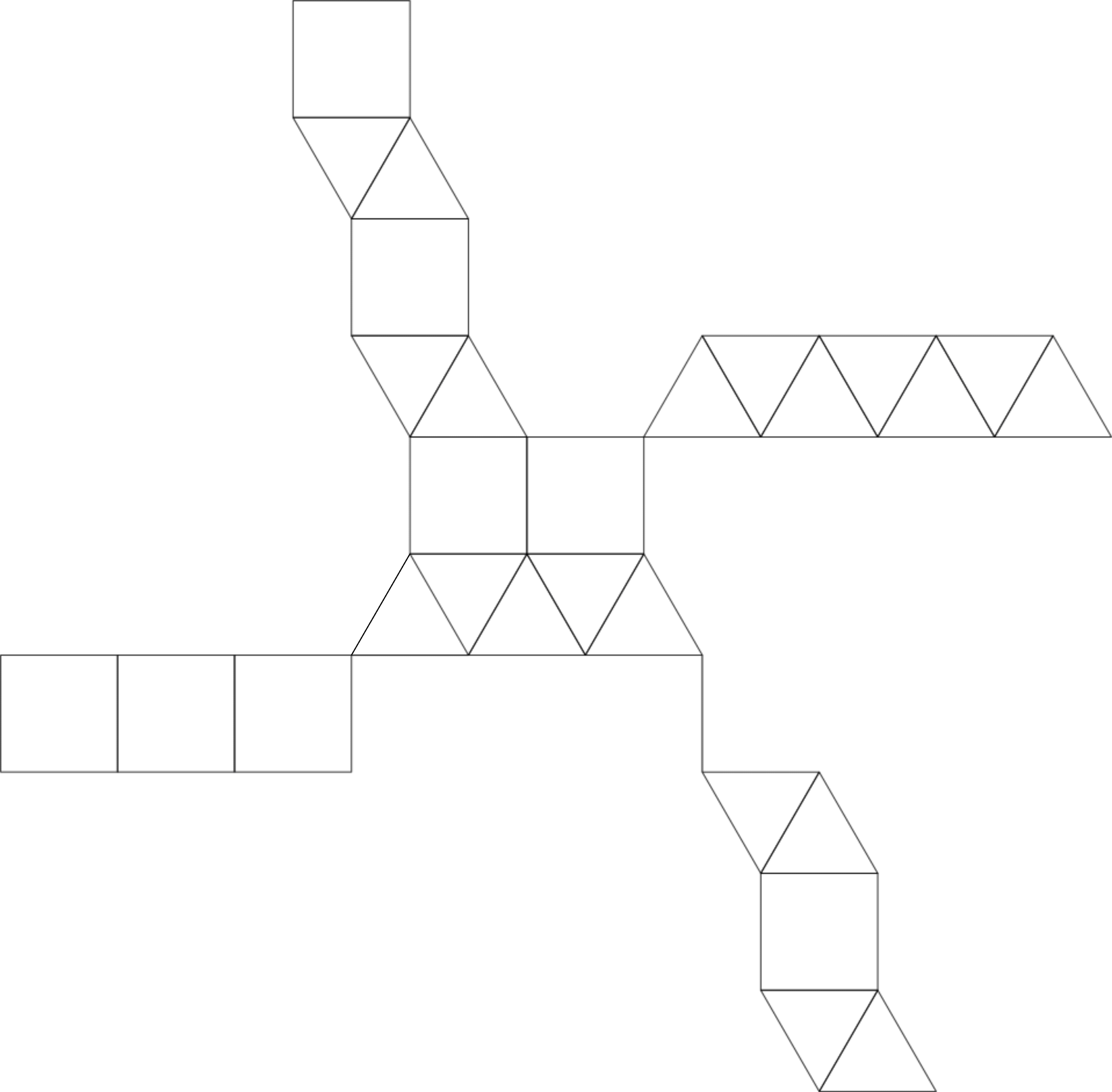}\\
\caption{An odd vertex gadget}
\label{fig:33344parrallel}
\end{figure}
Although one pin connection may be disabled in a odd vertex gadget, there remains three other functioning pin connections. Because the reduction only requires max degree three vertices, there are still ways to make the reduction work. We construct the simulated grid $G$ in the following way. Given a parity preserving square grid embedding of the original max degree three bipartite graph $G'$ as mentioned in the 1982 paper \cite{HCPsquare}, we enlarge the embedding by a factor of 3 so that any single segment is at least three segments long and the parities of the vertices are preserved. We then adjust the embedding by replacing every disabled pin connection with a functioning pin connection. Figure~\ref{fig:33344embed} shows that if the upper connection is disabled, we use the left or right connection to replace it (the upper row represents embedding before adjustment while the lower row represents embedding after adjustment). Based on the adjusted embedding, we can then construct a simulated graph $G$ using the square grid gadgets. Since the pin connections are all functioning in $G$, the reduction works.\\
\begin{figure}[h]
\centering
\includegraphics[width=80mm]{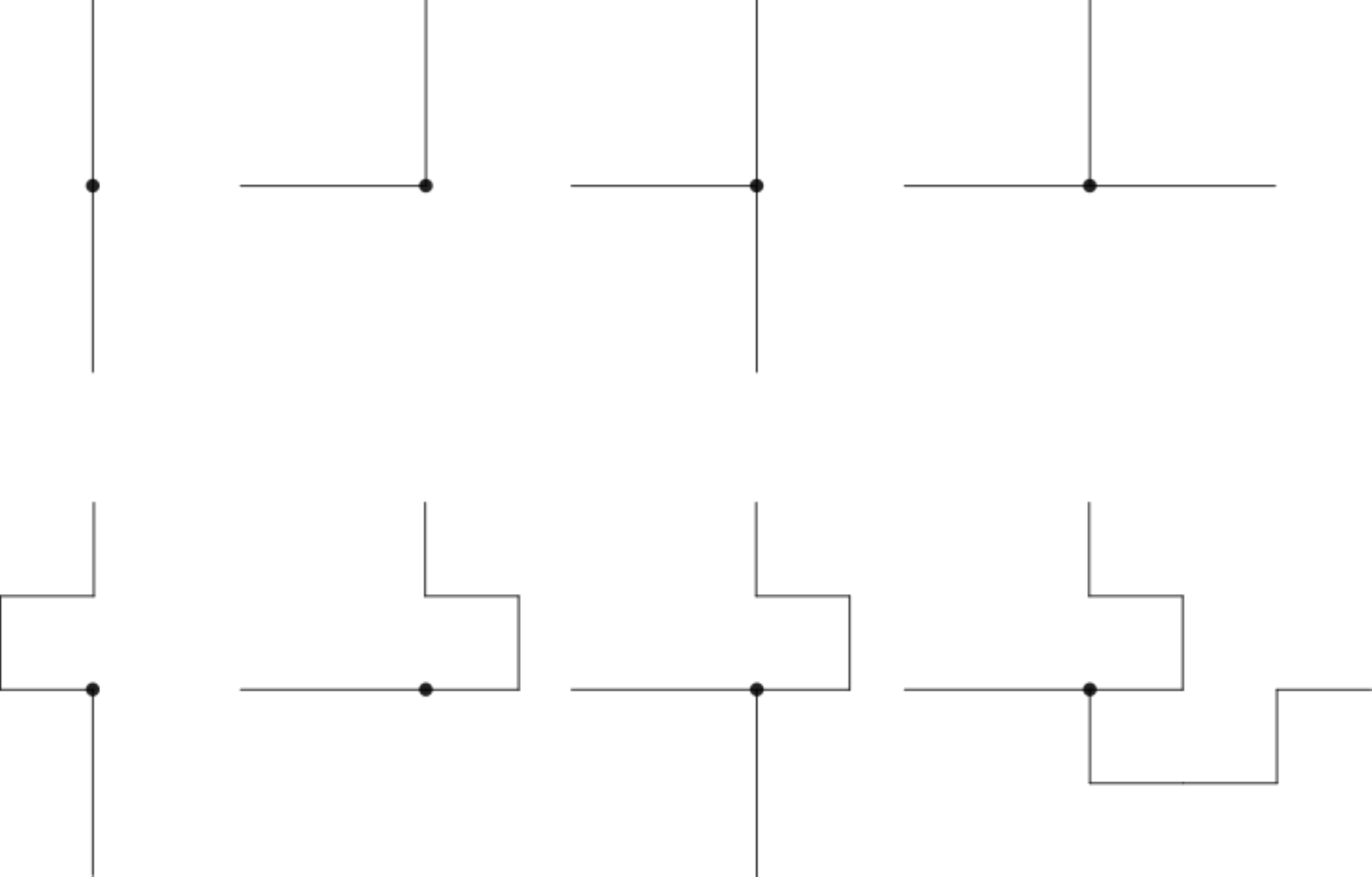}\\
\caption{Embedding adjustment}
\label{fig:33344embed}
\end{figure}
\end{proof}

\subsection{HCPs that Reduce from the Tree-Residue Vertex Breaking Problem}
\label{sec:TRVB}
In this section, we show the HCPs in the 4.8.8 tessellation and the 4.6.12 tessellation are NP-complete by reducing from the Tree-Residue Vertex Breaking Problem studied in\cite{TRVBresult}. Here, \emph{breaking} a degree n vertex means turning the vertex into n degree one vertices that are at the ends of the n edges. Tree-Residue Vertex Breaking problem asks that given a planar multigraph $M$ and with some of its vertices marked breakable, is it possible to break some of the breakable vertices so that the resulting graph is a tree. N-Regular Breakable Planar Tree-Residue Vertex-Breaking problem asks that given a planar multigraph with all the vertices degree n and breakable, is it possible to produce a tree from breaking some vertices. The HCPs in these section reduce from 4-Regular Breakable Planar Tree-Residue Vertex-Breaking problem and 6-Regular Breakable Planar Tree-Residue Vertex-Breaking problem, both of which are NP-complete\cite{TRVBresult}. The reduction works in this fashion: for any graph $M$, we will construct a grid graph $G$ of the tessellation so that $G$ has a Hamiltonian cycle if and only if $M$ is breakable.

\subsubsection{4.8.8 Tessellation}
\begin{figure}[H]
\centering
\includegraphics[width=70mm]{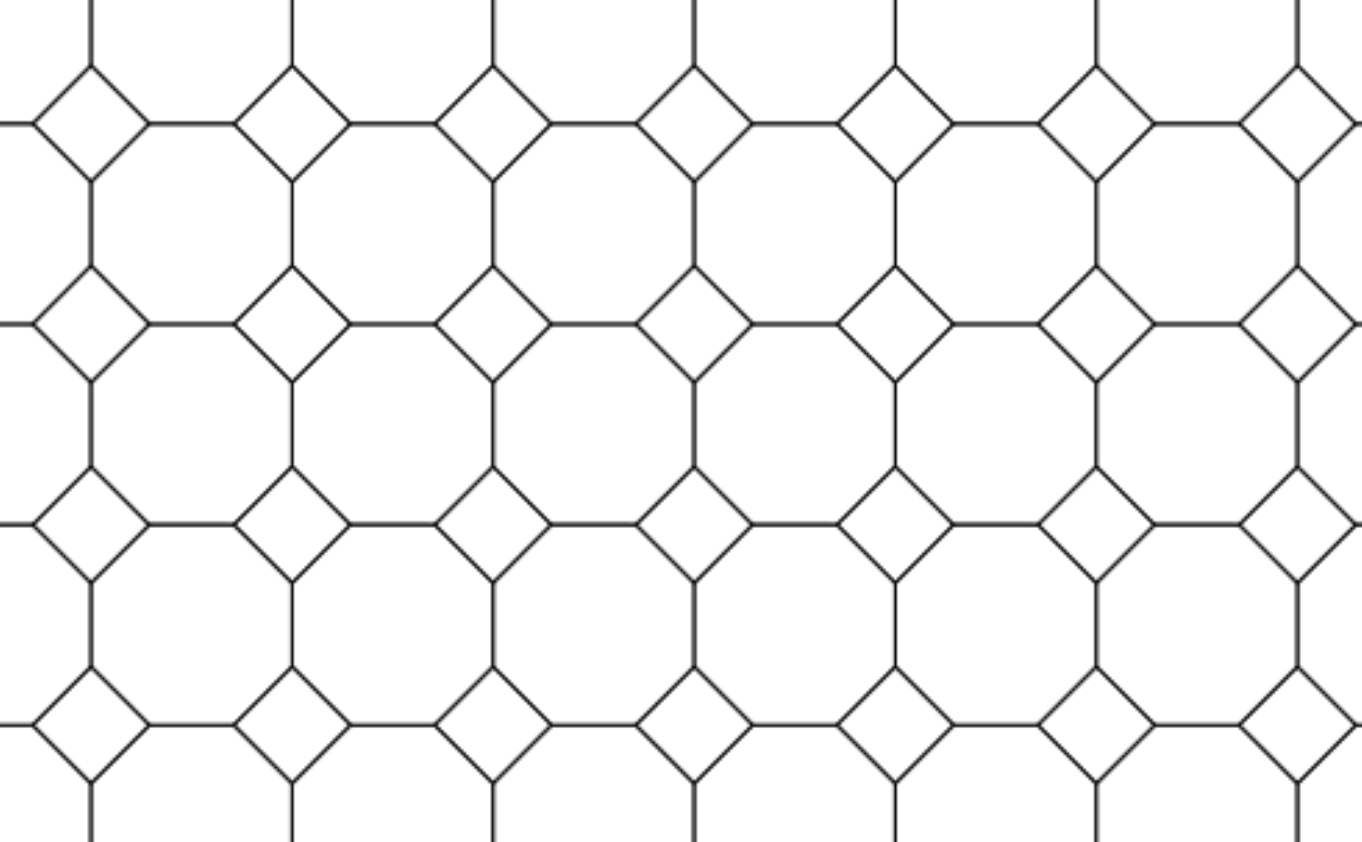}\\
\end{figure}
\begin{theorem}
The HCP in the grid graphs of the 4.8.8 tessellation is NP-complete.\\
\end{theorem}
\begin{proof}
We reduce from the 4-Regular Breakable Planar Tree-Residue Vertex-Breaking problem. When constructing a grid graph $G$ of the 4.8.8 tessellation based on $M$, we first make a square grid embedding of $M$, using a method such as the one described in\cite{Schaffter-1995}. Then, for each vertex of $M$, we use the vertex gadget in Figure~\ref{fig:488vertex}. For the edges in the embedding, we use the edge gadget formed by the boundary vertices of a three-octagon wide strip, as shown in Figure~\ref{fig:488wire}. Notice that the edge gadget can shift and turn easily. Due to this flexibility, we can form a graph $G$ based on the embedding using the gadgets.\\
\begin{figure}[H]
\centering
\includegraphics[width=80mm]{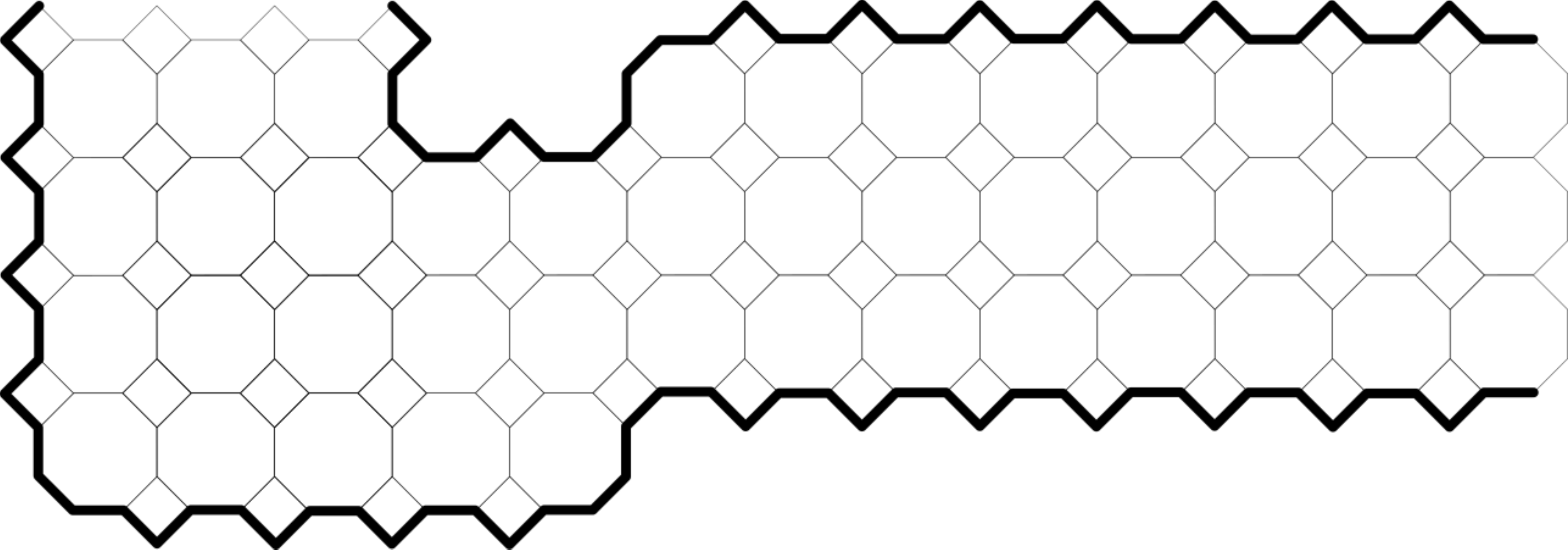}\\
\caption{Edge gadget with a turn for 4.8.8}
\label{fig:488wire}
\end{figure}
\begin{figure}[H]
\centering
\includegraphics[width=60mm]{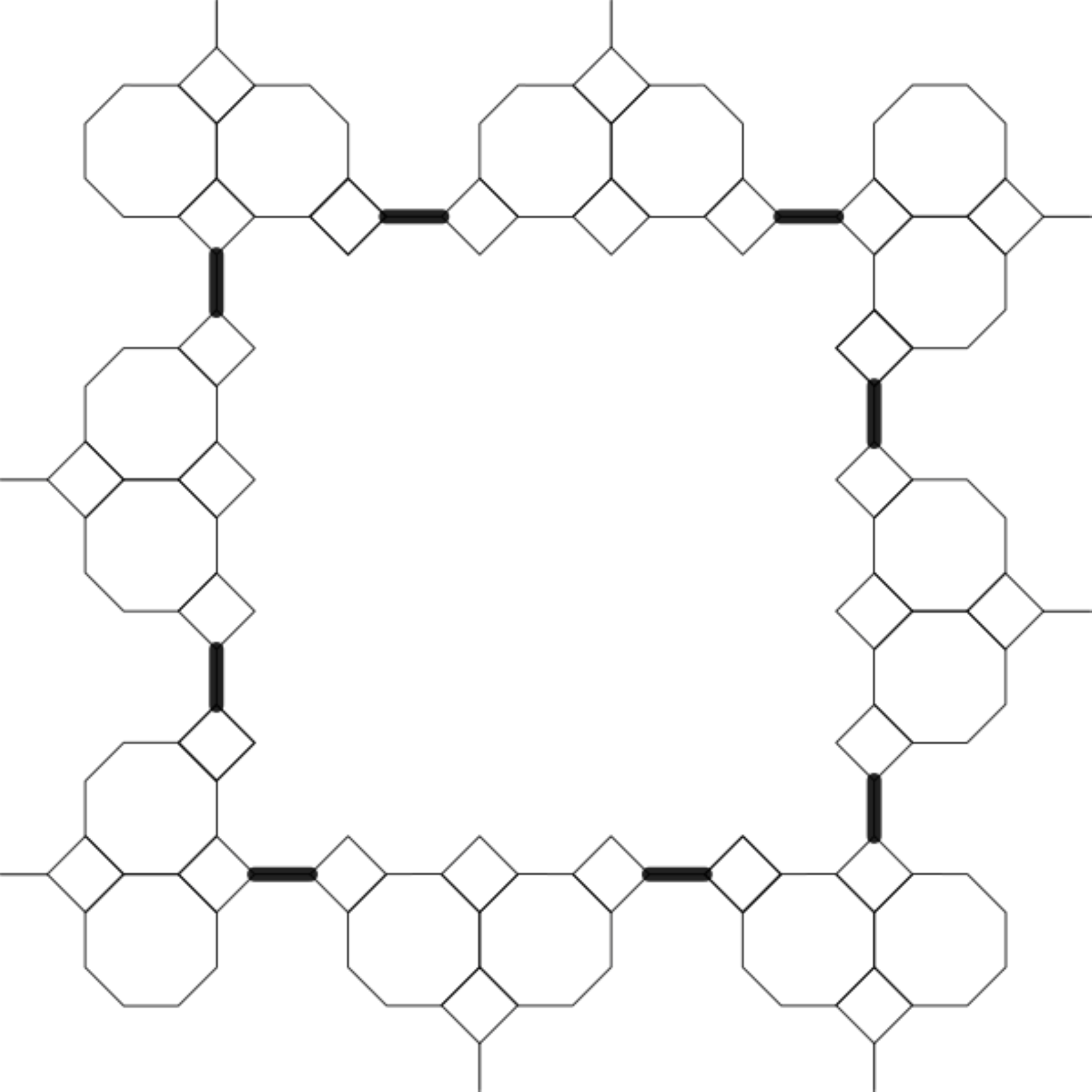}\\
\caption{Vertex gadget for 4.8.8}
\label{fig:488vertex}
\end{figure}
Now, we will show the constructed graph $G$ has a Hamiltonian cycle if and only if $M$ is breakable. Noticed that if $G$ has a cycle $C$, all the sides on the edge gadgets must be in $C$, and the freedom is only in how to traverse the vertex gadgets. Figure~\ref{fig:488vertexsolution} shows two solution to the vertex gadget. The four edge gadgets connect to the vertex on the four sides of it. Each edge gadgets has two separate paths of vertices that go into the vertex gadget. Note that there are eight single connection edges (bold edges in Figure~\ref{fig:488vertex}) in the vertex gadgets, each of which is in between a pair of adjacent series. If a cycle exists and a path comes in from a string, the path must enter one of the two adjacent single edge connections and then connect with the path coming in from another string. Thus, for a vertex gadget, there are only two kinds of solution: one that has two strings of the same edge connected or one that has two strings of two adjacent edges connected. The first kind is illustrated by the solution on the left, which correspond to a broken vertex in $M$ while the second kind is illustrated by the solution on the right side, which correspond to a unbroken vertex in $M$. To show that $G$ has a Hamiltonian cycle if and only if $M$ is breakable, we apply the reasoning used in the 2017 paper\cite{TRVBreduction}. If $M$ is breakable, then for every broken vertex in $M$, we traverse through the corresponding vertex gadget using the broken solution; for every unbroken vertex, we traverse through the corresponding vertex gadget using the unbroken solution. Note that after this procedure, the graph produced by breaking $M$ is the same as the region inside the edge gadgets in $G$. If the graph produced by breaking $M$ is indeed a tree, which is connected and acyclic, then the region inside the edges must also be connected and hole-free, which shows that there is a Hamiltonian cycle. If there is a Hamiltonian cycle in $G$, the region inside must by connected and hole-free, which then show that the graph $M$ can be broken down to a tree.
\begin{figure}[h]
\centering
\includegraphics[width=80mm]{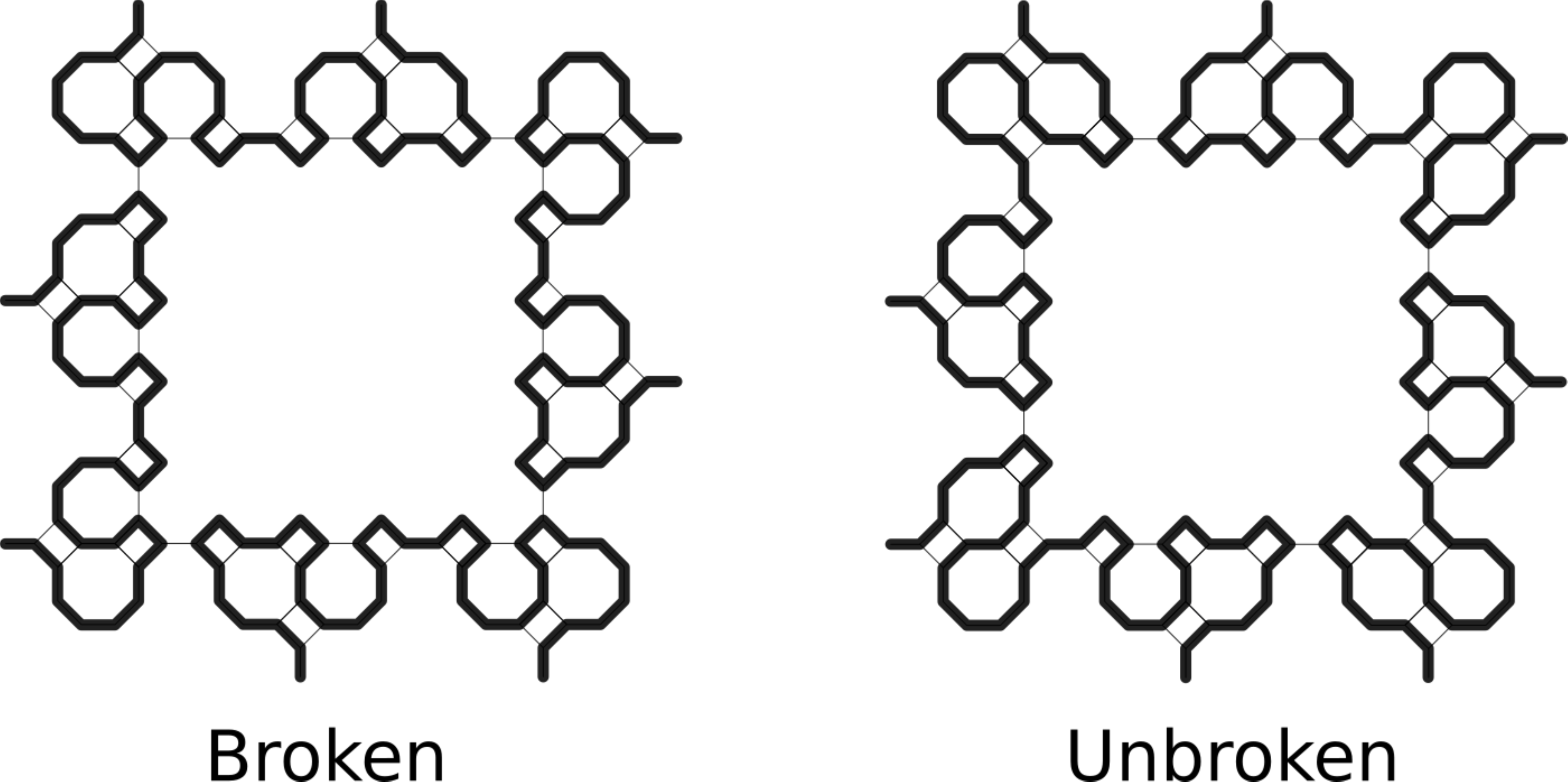}
\caption{Solutions for the 4.8.8 vertex gadget}
\label{fig:488vertexsolution}
\end{figure}
\end{proof}

\subsubsection{4.6.12 Tessellation}
\begin{figure}[H]
\centering
\includegraphics[width=70mm]{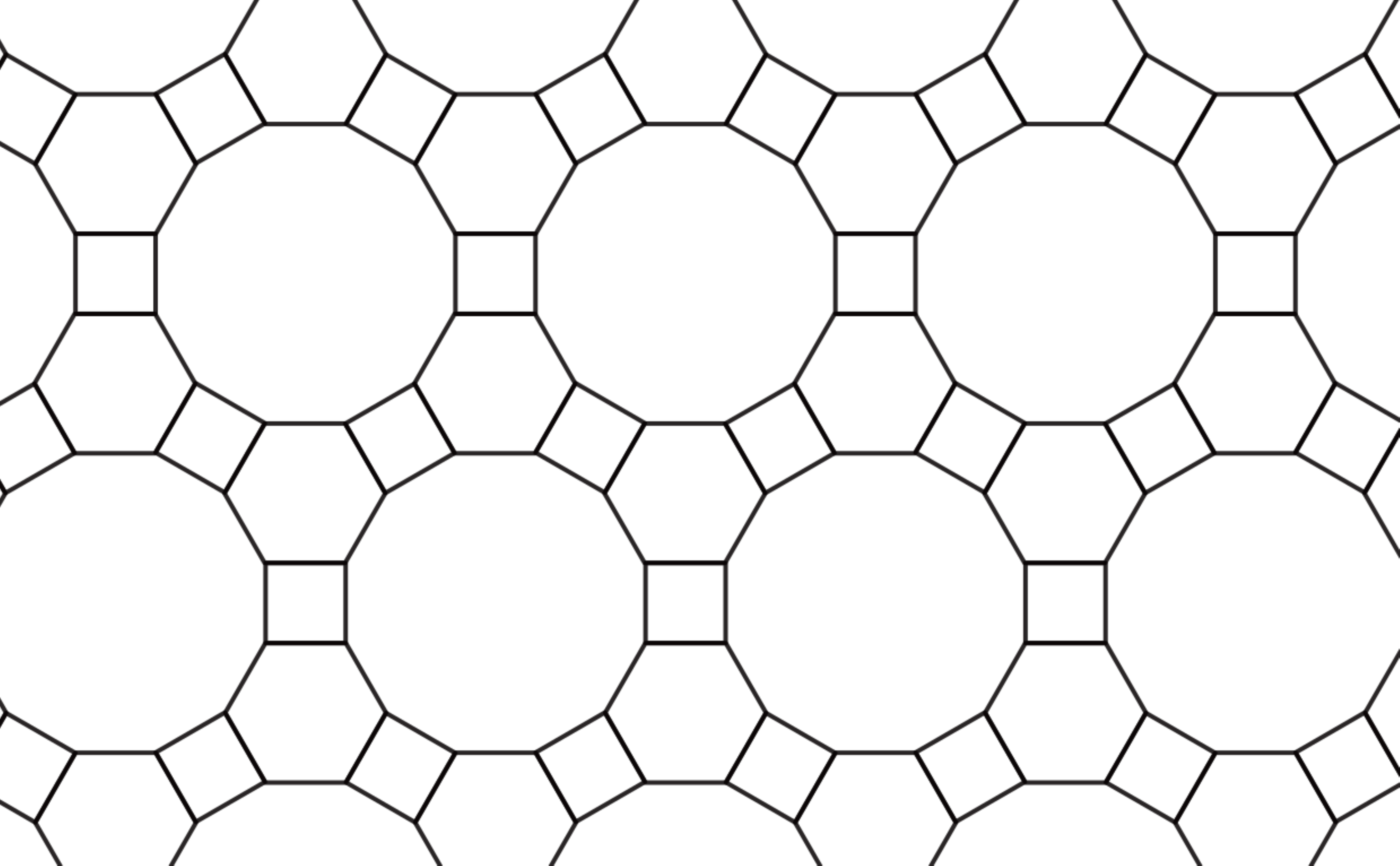}
\end{figure}
\begin{theorem}
The HCP in the grid graphs of the 4.6.12 tessellation is NP-complete.\\
\end{theorem}
\begin{proof}
We prove that the HCP in 4.6.12 Tessellation is NP-complete by reducing from the 6-Regular Breakable Planar Tree-Residue Vertex-Breaking problem. When constructing a grid graph $G$ in 4.8.8 tessellation based on the multigraph $M$, we first embed the multigraph in the triangular grid. Then, we use the vertex gadget shown in Figure~\ref{fig:4612vertex} for every vertex in $M$ and the edge gadget shown in Figure~\ref{fig:4612edge} for the edges in $M$. The edge gadget only includes the boundary vertices of the shape depicted in Figure~\ref{fig:4612edge}. Because the turning demonstrated in \ref{fig:4612edge} can have turning of $60$ and $120$ degrees, we can construct the induced subgraph $G$ based on the triangular grid embedding.\\
\begin{figure}[H]
\centering
\includegraphics[width=80mm]{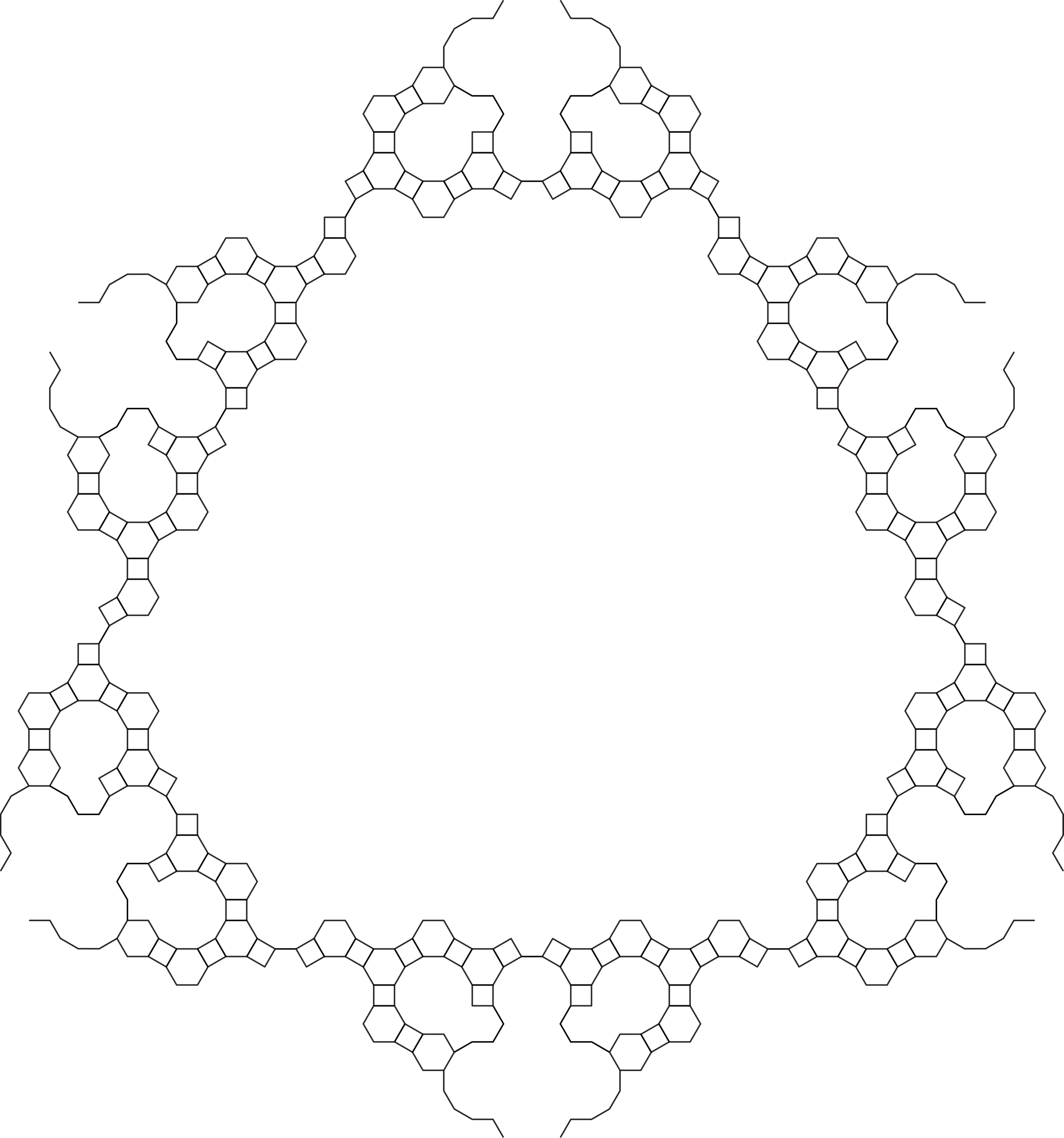}
\caption{Vertex gadget for 4.6.12}
\label{fig:4612vertex}
\end{figure}
\begin{figure}[H]
\centering
\includegraphics[width=80mm]{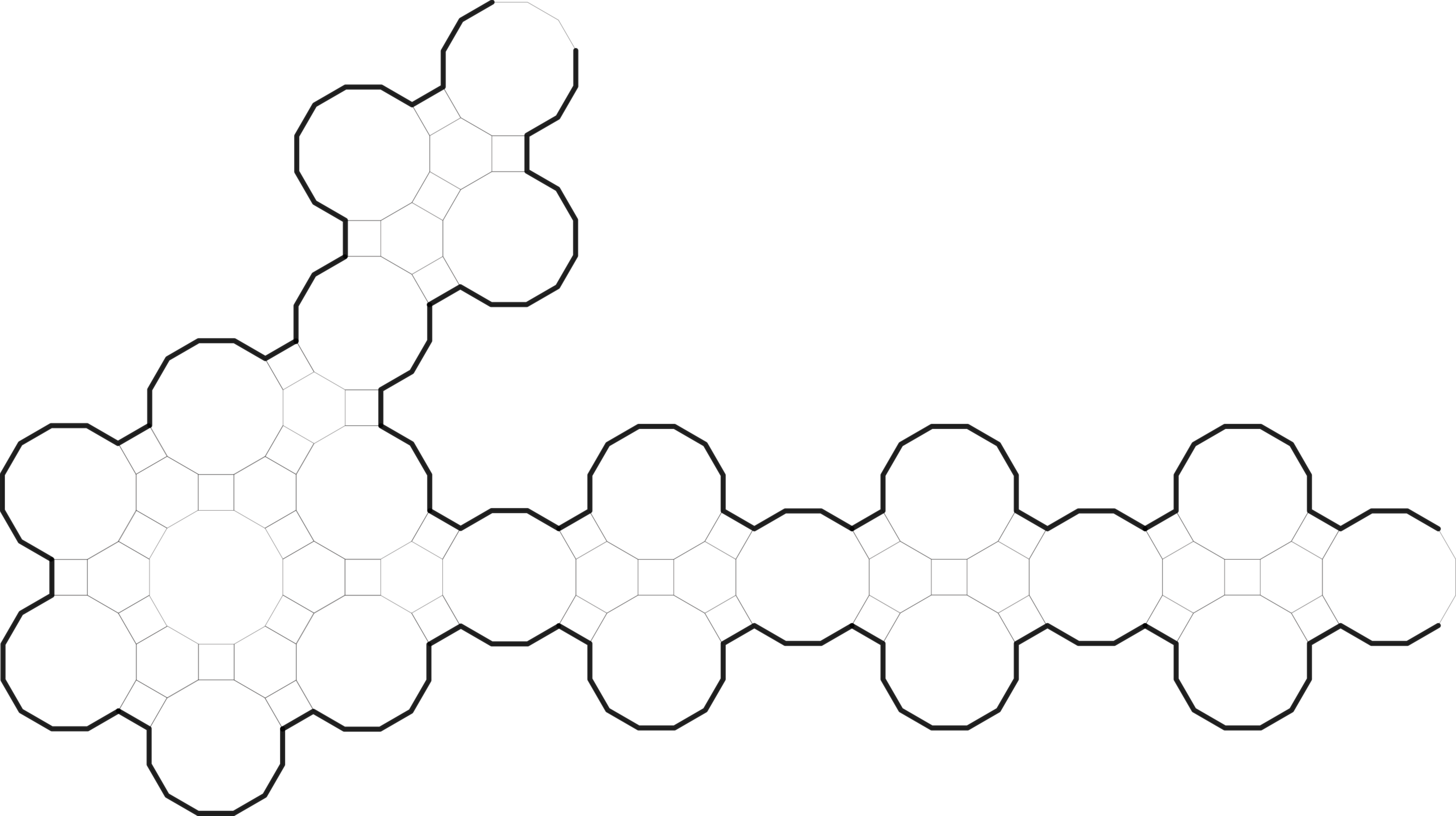}
\caption{Edge gadget with a turn for 4.6.12}
\label{fig:4612edge}
\end{figure}
Now, we will show why the constructed graph $G$ has a Hamiltonian cycle if and only if $M$ is breakable. The traversals of the edge gadgets of 4.6.12 tessellation are already set and the only freedom is in how to traverse the vertex gadgets. The six edge gadgets connect to the vertex gadget on the six sides and each edge gadget consists of two strings of vertices. As mentioned in the 4.8.8 tessellation, because of the single edge connections between each pair of adjacent strings, there are only two kinds of traversals for a vertex gadget: the one that has two strings of the same edge connected or the one that has two strings of two adjacent edges connected. The first kind is illustrated by the solution in Figure~\ref{fig:4612break}, which corresponds to a broken vertex in $M$. The second kind is illustrated by the solution in Figure~\ref{fig:4612notbreak}, which corresponds to an unbroken vertex in $M$. Just as the argument in 4.8.8 tessellation proof states, the region inside the edge gadgets represents the produced graph after breaking $M$.
\begin{figure}[H]
\minipage{0.5\textwidth}
\includegraphics[width=58mm]{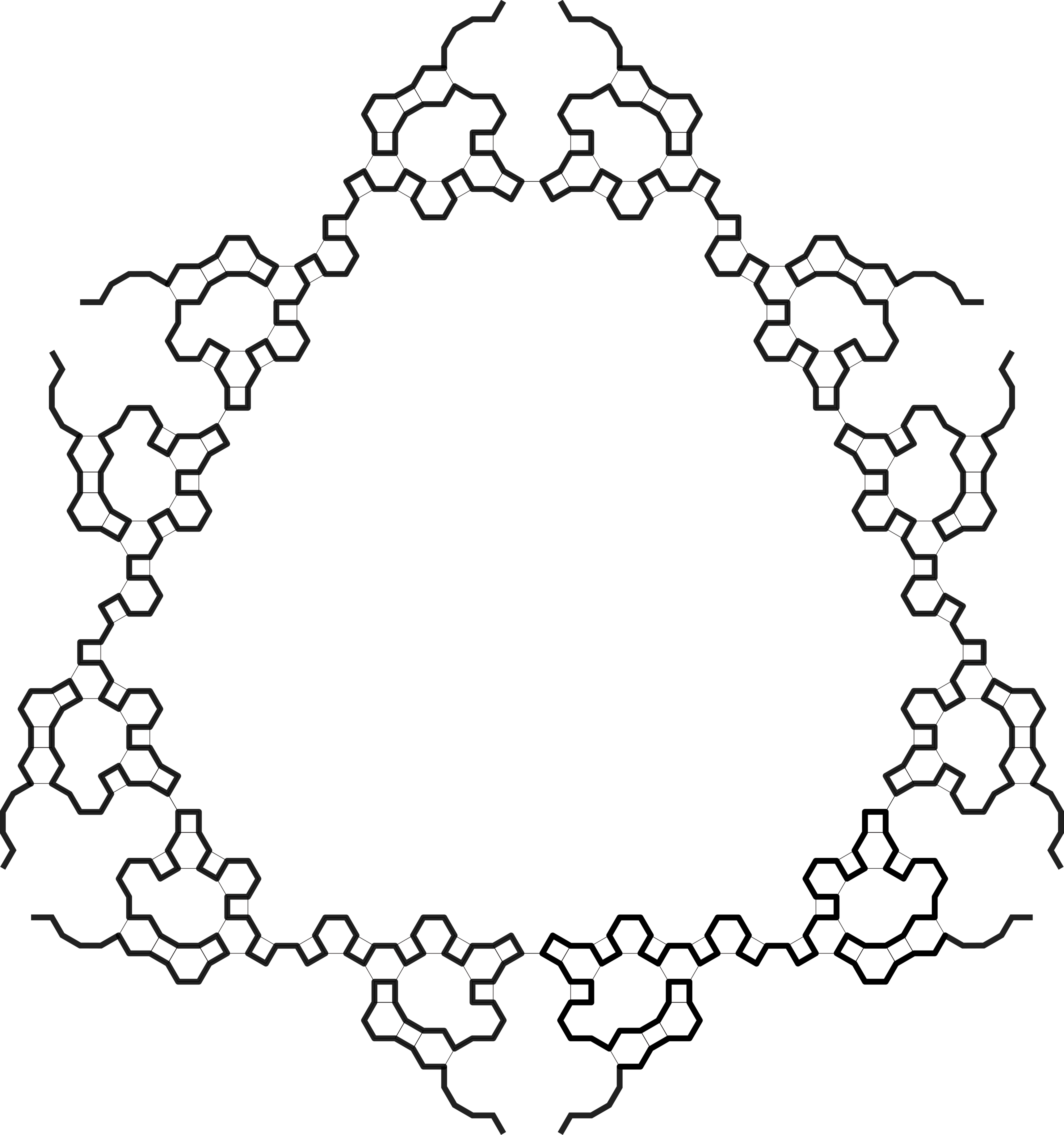}
\caption{A Broken Vertex}
\label{fig:4612break}
\endminipage
\minipage{0.5\textwidth}
\includegraphics[width=58mm]{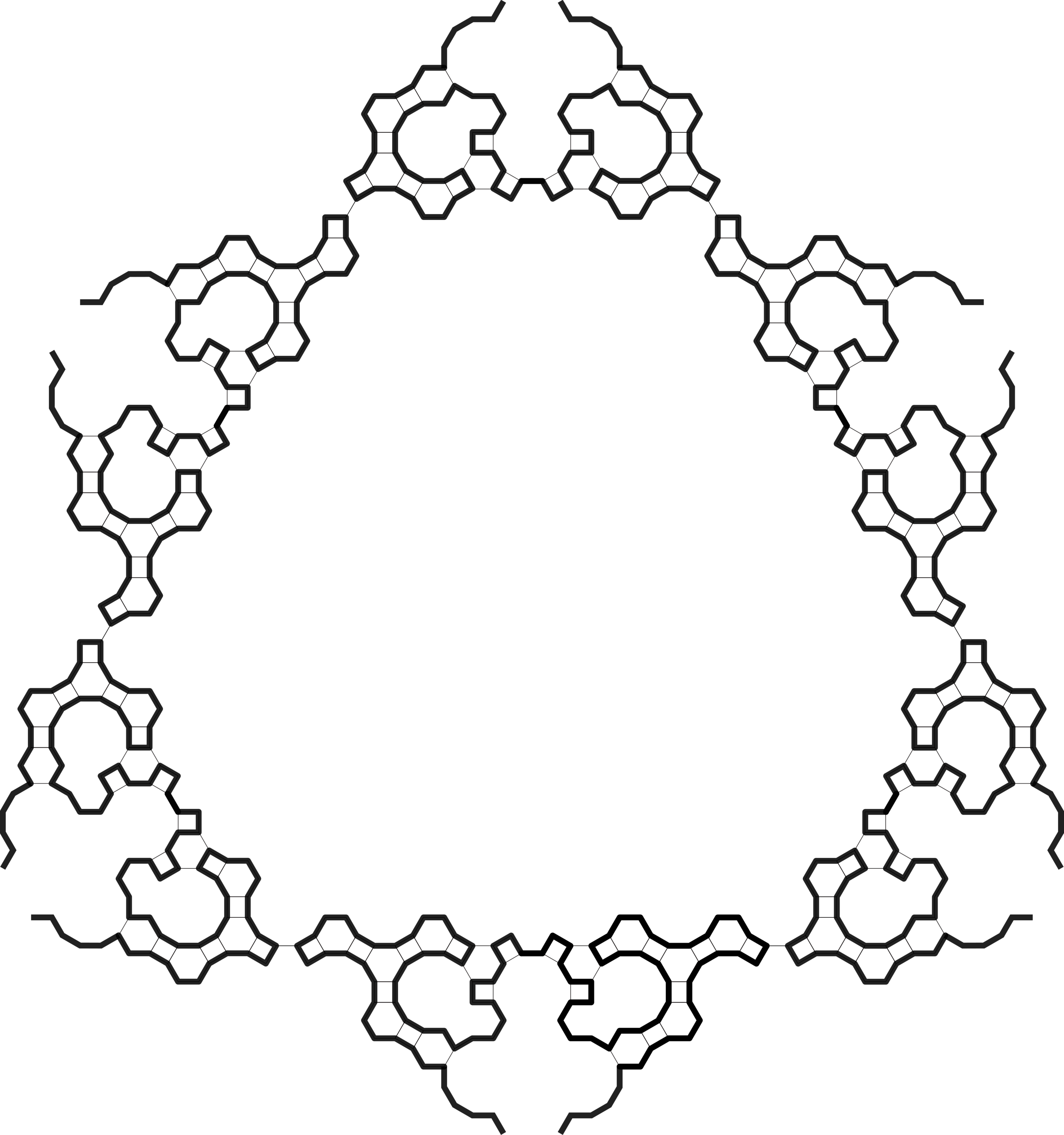}
\caption{A Unbroken Vertex}
\label{fig:4612notbreak}
\endminipage
\end{figure}
\end{proof}

\section{Hamiltonian Path with Turns}
\label{sec:turning}
In this section we explore whether grid graphs contain Hamiltonian cycles which turn at every vertex. In \cite{fekete1997angle}, Fekete and Woeginger give a near linear algorithm for finding Hamiltonian paths among a set of points in the plane when the path must turn by $90^\circ$ at every vertex. We show that this angle-restricted tour problem becomes NP-complete when generalized to 3D, even when we restrict to 3D square grid graphs whose height is only two vertices. We also show the complexity of finding always turning Hamiltonian cycles in triangular grids.

We also investigate a version where each vertex can be visited twice as long as no edges in the cycle overlap. We give linear time algorithms for finding always turning cycles, as well as double visiting cycles in solid square grid graphs. This question initially came to our attention as special cases of finding Hamiltonian cycles in the squares and octagons tessellation. There are clear reductions between various problems in this section and restricted versions of that problem in which all vertices around a square pixel are included if any one is included. Although this did not lead to our eventual hardness proof we found the problem to be interesting and well motivated on its own. One can look at this problem as mirroring a problem laid out in a grid where movement is reflected by barriers at 45 degree angles. One can also think of this as counterpart to the discrete milling problem, were turns happen to be much easier to perform than straight paths.

\subsection{Turning in 3D Square Grids is Hard}
\label{sec:cubicGrid}

The 3D square grid graph is the grid graph induced by the integer lattice in $\mathbb{R}^3$. It is a natural extension of the square grid, and thus it makes sense to see the impact of the turning restriction in this setting. Interestingly, the problem of finding a Hamiltonian cycle once again becomes NP-complete, even when the cubic grid is restricted to being two layers tall. Thus, only a little more maneuverability is needed for this problem to once again become hard.

Turning is once again defined to match the intuitive notion. Each vertex now contains three pairs of opposite potential edges, and a turn in a path cannot contain both edges in the pair.

\begin{theorem}
The Always Turning Hamiltonian Path problem in cubic grid graphs is NP-complete even if the height of the grid is restricted to be $2$ vertices.
\end{theorem}
\begin{proof}
We closely follow the proof that deciding if square grid graphs admit a Hamiltonian path is NP-complete. We reduce from deciding whether planar max-degree 3 graphs admit a Hamiltonian path. We also construct edge gadgets and even and odd vertex gadgets. In this subsection, all figures are two vertices high with paths on the bottom layer represented by black lines and paths on the top layer represented by dotted blue lines.

Edge gadgets are sequences of $2\times2\times 2$ cubes. They admit a forward path, representing an edge taken in the graph, shown in Figure~\ref{fig:EdgePathThrough}. They also admit a return path, shown in Figure~\ref{fig:EdgePathReturn}, representing an edge not taken in the graph. 
The edges can be turned, as shown in Figures~\ref{fig:TurnThrough} and \ref{fig:TurnReturn}.

\begin{figure}[!tbp]
  \centering
  \begin{minipage}[b]{0.4\textwidth}
    \includegraphics[width=\textwidth]{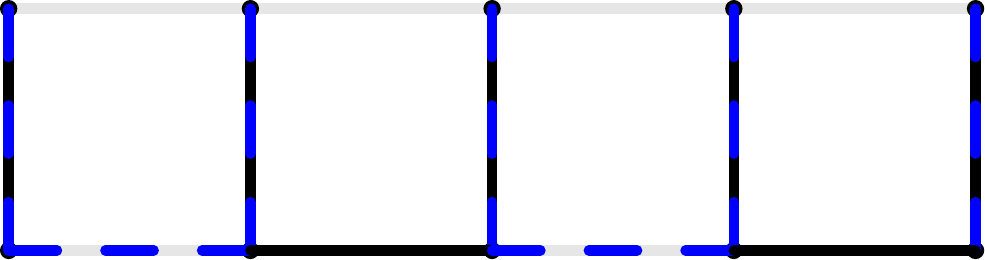}
    \caption{An edge taken in the simulated graph. The path here starts on one side and ends on the other}
    \label{fig:EdgePathThrough}
  \end{minipage}
  \hfill
  \begin{minipage}[b]{0.4\textwidth}
    \includegraphics[width=\textwidth]{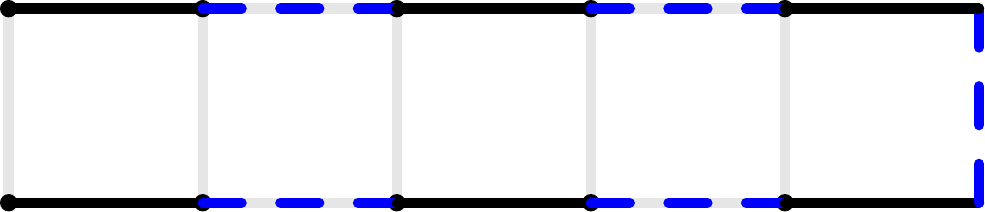}
    \caption{An edge not taken in the simulated graph. The path starts and ends both on the left side.}
    \label{fig:EdgePathReturn}
  \end{minipage}
  \\
    \begin{minipage}[b]{0.4\textwidth}
    \includegraphics[width=\textwidth]{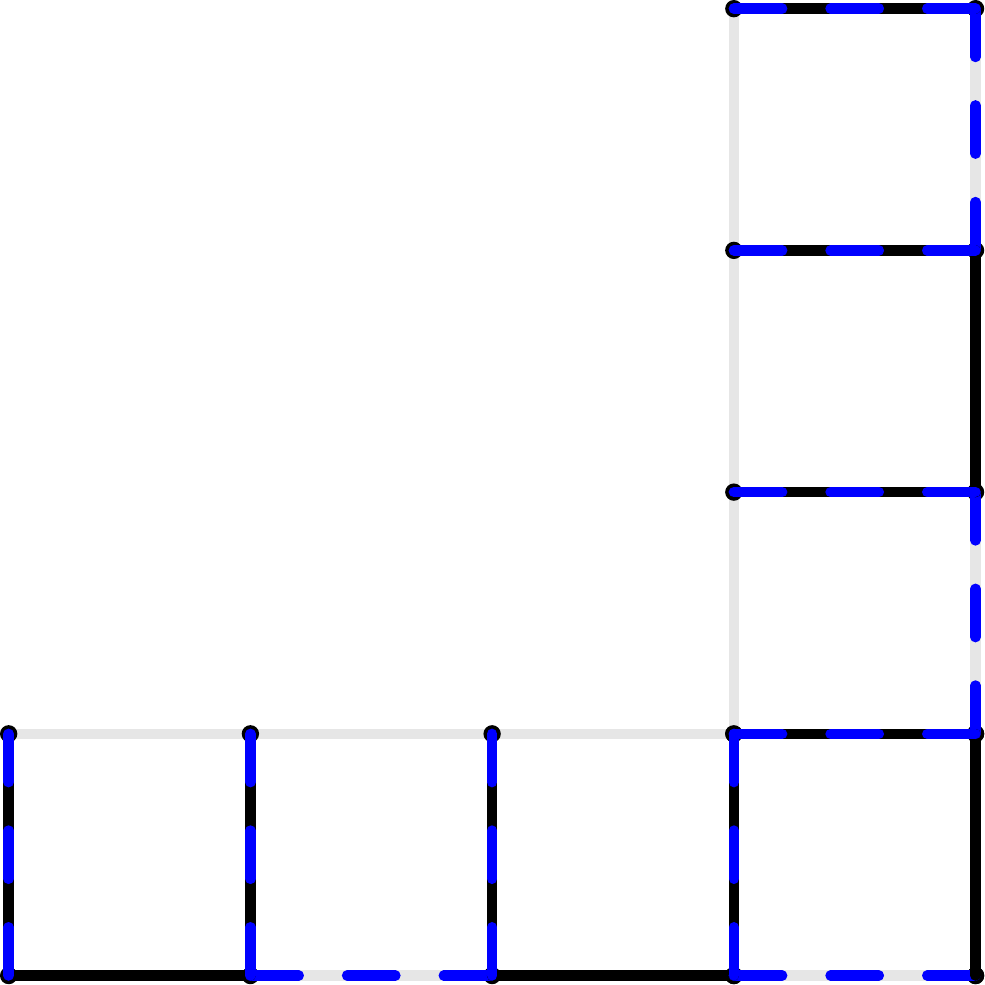}
    \caption{A taken edge being turned.}
    \label{fig:TurnThrough}
  \end{minipage}
  \hfill
  \begin{minipage}[b]{0.4\textwidth}
    \includegraphics[width=\textwidth]{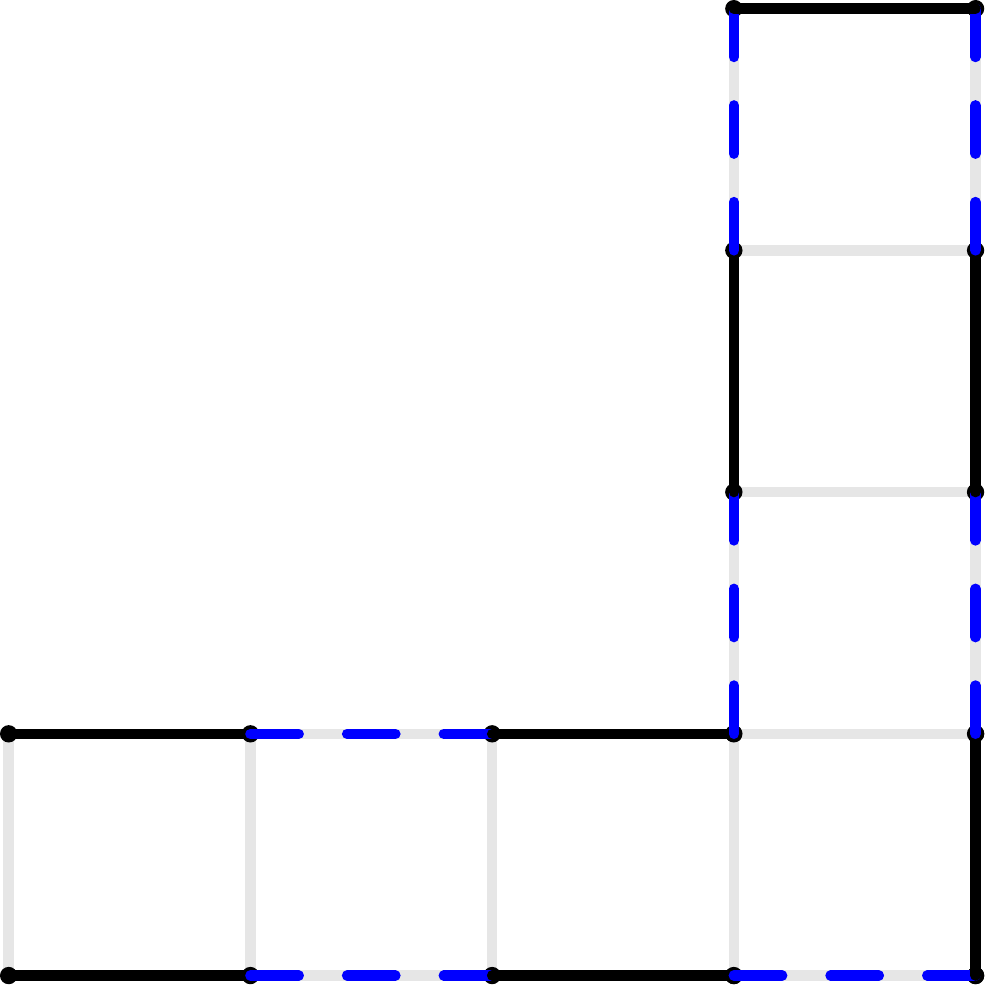}
    \caption{A non-taken edge being turned.}
    \label{fig:TurnReturn}
  \end{minipage}
\end{figure}

Vertex gadgets are represented by $4\times8\times2$ rectangles. Edges attach across the marked edges $e_1$ to $e_4$. Figures~\ref{fig:VertexPath1}, \ref{fig:VertexPath2}, and \ref{fig:VertexPath3} show three different paths through a vertex which will connect any three of the four target edges. Unlike the original grid proof, it is critical that the problem we are reducing from is max-degree 3.  Even vertex gadgets attach to the edge gadgets by a $1\times2\times1$ pair of vertices. If a path enters this pair of vertices, it is then forced to take the edge connecting them. Thus in an even vertex the path can only pass from the vertex to each edge a single time. Since it is max-degree 3, this means precisely two of the adjacent edge gadgets are taken and one is not. Since every odd vertex is connected to an even vertex, this means the odd vertex gadgets must also have precisely two of their adjacent edge gadgets have a taken path. Thus there is only a Hamiltonian cycle if the original graph admits a Hamiltonian cycle.

\begin{figure}[!tbp]
  \centering
  \begin{minipage}[b]{0.3\textwidth}
    \includegraphics[width=\textwidth]{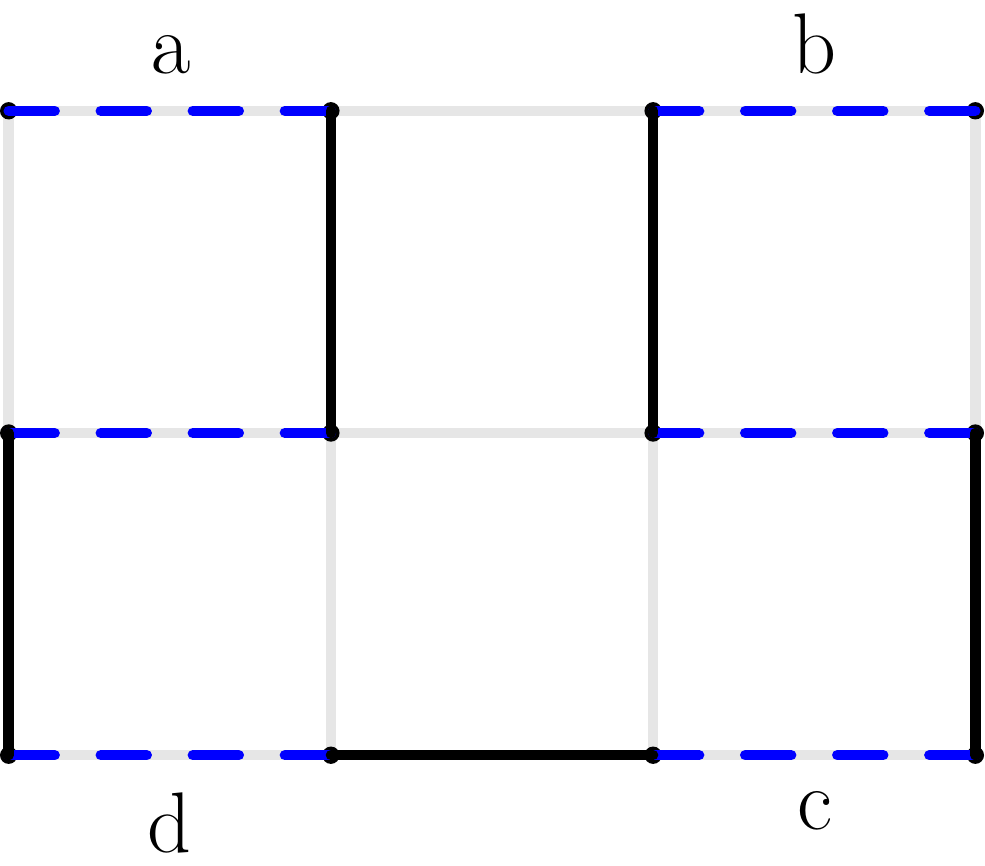}
    \caption{Enters a, crosses c, leaves b.}
    \label{fig:VertexPath1}
  \end{minipage}
  \hfill
  \begin{minipage}[b]{0.3\textwidth}
    \includegraphics[width=\textwidth]{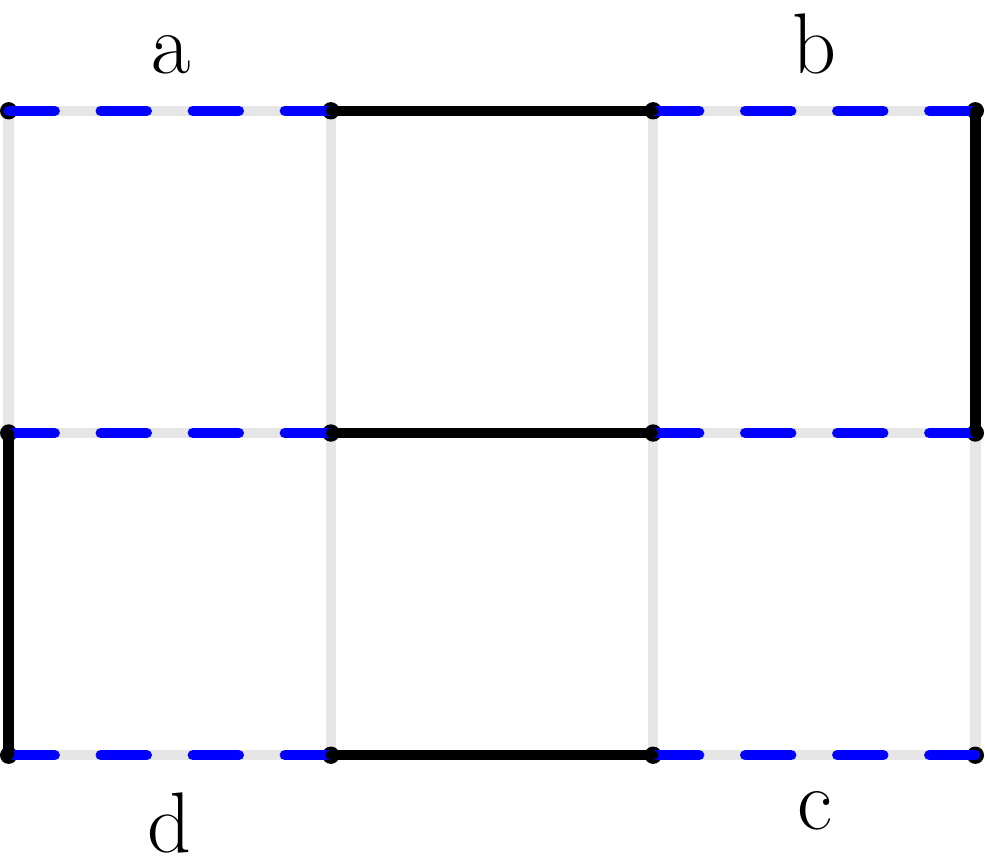}
    \caption{Enters a, crosses b, leaves c}
    \label{fig:VertexPath2}
  \end{minipage}
  \hfill
    \begin{minipage}[b]{0.3\textwidth}
    \includegraphics[width=\textwidth]{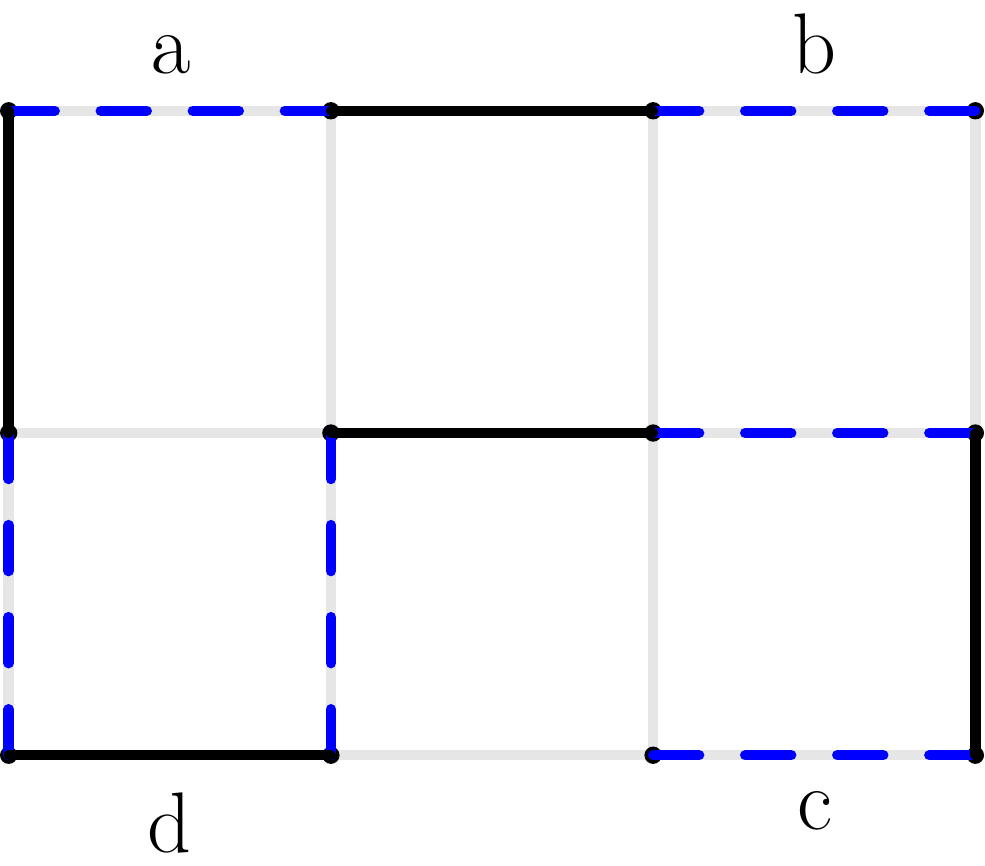}
    \caption{Enters b, crosses a, leaves c.}
    \label{fig:VertexPath3}
  \end{minipage}
\end{figure}

\end{proof}

\subsection{Double Turning in Solid Graphs}
\label{sec:doubleTurning}
Initially inspired by the Octagon and Squares we consider the following problem. We define the Double Turning Hamiltonian Cycle Problem to be the following: Given a square grid graph, does there exist a cycle in that graph which visits every vertex at least once, never traverses an edge more than once, and turns at every vertex? In particular, this allows the path to visit degree 4 vertices twice, taking a different turn each time. If we consider the square and octagon tessellation in which we have every vertex around a square pixel if any vertex is present around that pixel, then one can see these are equivalent problems.

This section will begin by observing some useful properties of the Hamiltonian path. Then we will connect those to properties of the graph to show that these graphs have a property we call a checkering. Next, we demonstrate that spanning trees of the checkering correspond to Hamiltonian cycles in our graph. Finally, we argue that all of these properties can be checked in polynomial time.

\begin{lemma}
\label{lem:allFour}
If a solid graph admits a double turning Hamiltonian cycle, it also admits such a cycle where all degree 4 nodes are visited twice.
\end{lemma}
\begin{proof}
If a degree 4 node has only one visit then two adjacent edges must be in the path and two adjacent edges must not be in the path. Let us consider some properties of the empty edges. Since each edge must have a partner in the node, then each `path' of empty edges must either connect to degree 3 nodes or be in a cycle. Degree 3 nodes only occur on the boundary since this is solid, so if we have a path from one degree 3 node to another, it has disconnected part of our graph and thus cannot be part of a valid solution. Thus the empty edges must form a cycle. If this cycle contains any nodes on its interior, then those nodes are disconnected and the cycle cannot be part of a valid solution. Finally a node cannot be visited by an empty path more than once, otherwise it is never visited in the actual path. The only cycles in a grid which obey these constraints are single pixels. If there is a pixel without edges, then we pick one vertex arbitrarily to extend the cycle into the pixel. This is shown in Figures~\ref{fig:missingPixel} and \ref{fig:addedPixel}.
\end{proof}

\begin{figure}[h]
\begin{minipage}[b]{.45\textwidth}
\centering
	\includegraphics[width=.9\textwidth]{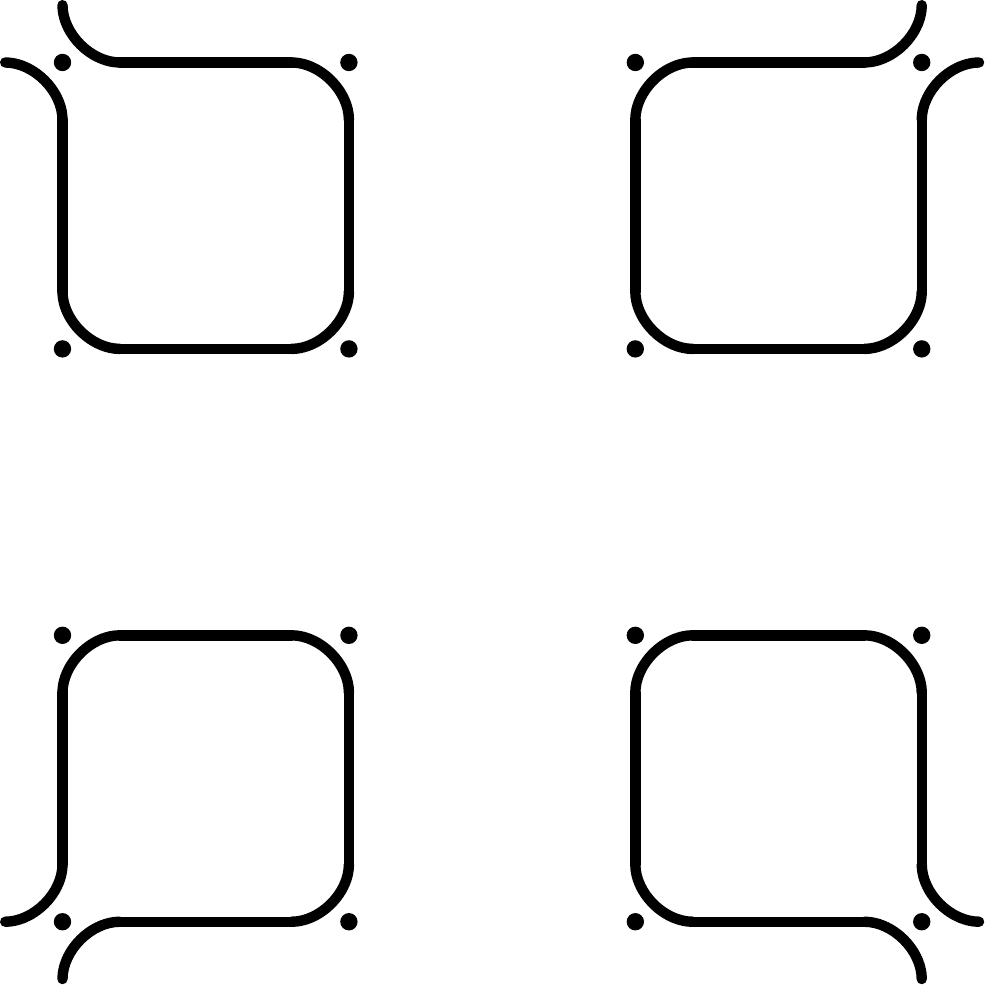}
	\caption{A pixel with all four edges not included in the path.}
	\label{fig:missingPixel}
\end{minipage}
\begin{minipage}[b]{.45\textwidth}
\centering
	\includegraphics[width=.9\textwidth]{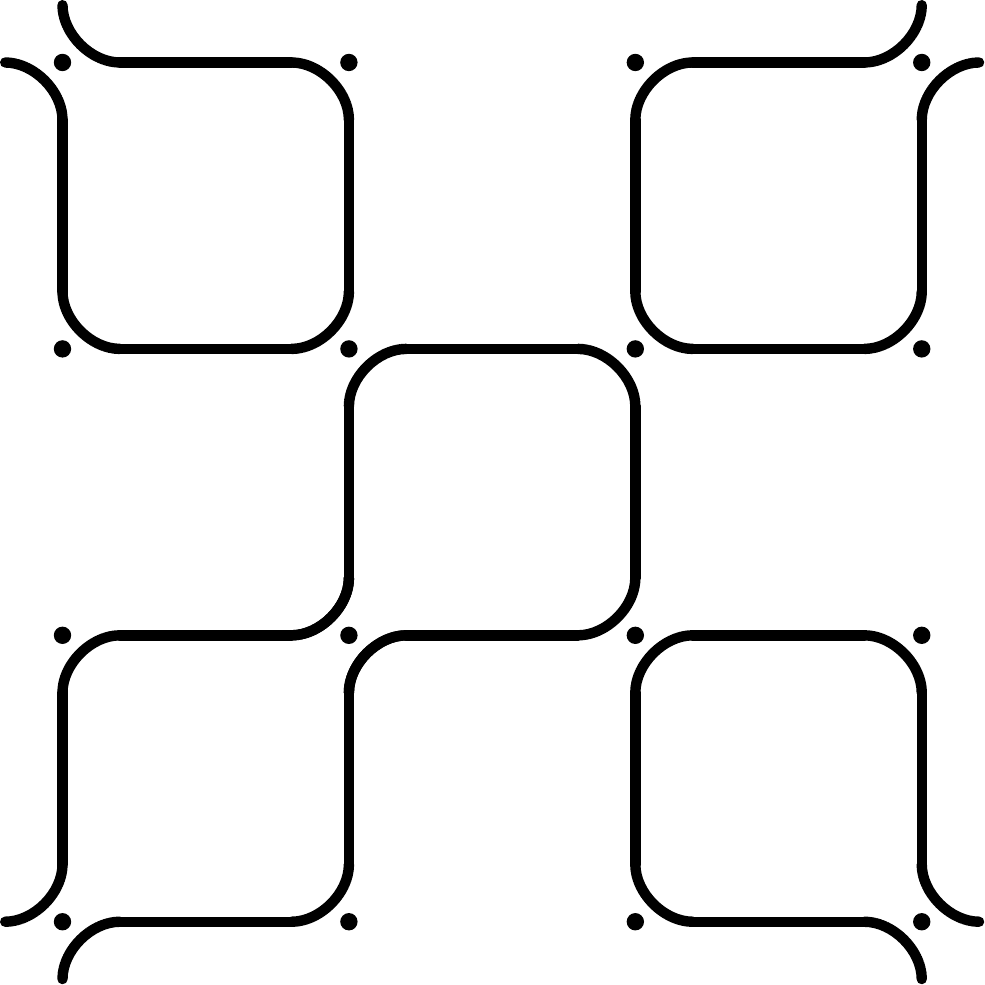}
	\caption{We can augment the path to visit each vertex twice..}
	\label{fig:addedPixel}
\end{minipage}
\end{figure}

With this lemma we can now restrict our examination to the case where all degree four nodes are visited twice. 

Next, we notice that the graph must have even parity on all external boundaries by the argument in the single turning place. Given this parity constraint and that the graph is solid, we know that if the graph contains a Hamiltonian cycle then it is composed of some number of full pixels, possibly connected at the corners\footnote{These corner connections are local cuts and what prevent this graph from being categorized as polygonal.}.
We now wish to consider an alternate view of this grid graph. Call the \emph{checkerboard} of this graph the set of alternating pixels in the graph starting with the upper left. We call the other pixels the odd checkering. From the prior properties, we know that if the graph contains a Hamiltonian cycle, then every vertex boarders at least one pixel in the checkerboard. 

Now we will imagine connecting the pixels in the checkerboard and show that the existence of a Hamiltonian cycle depends on its properties. Consider the degree 4 vertices, all of which are visited twice by our prior lemma. There are two configurations of paths, each one connects two diagonally adjacent pixels and separates the other two. We can now think of every degree 4 vertex of our graph as either connecting two adjacent checkered pixels or two adjacent odd checkered pixels. Our next lemma will consider this connectivity graph of the checkering.

\begin{figure}[h]
\centering
	\includegraphics[width=.8\textwidth]{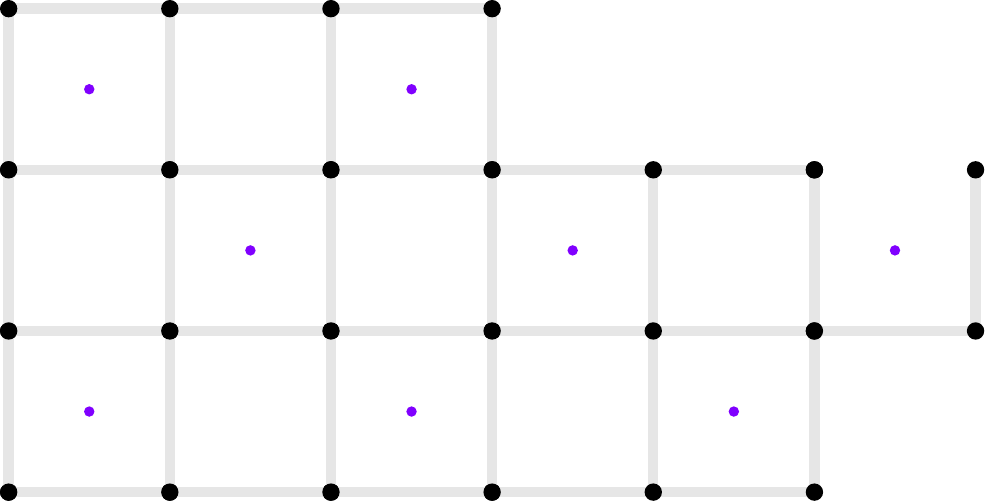}
	\caption{An example solid grid graph with its checkering in purple.}
	\label{fig:exampleCheckering}
	\includegraphics[width=.8\textwidth]{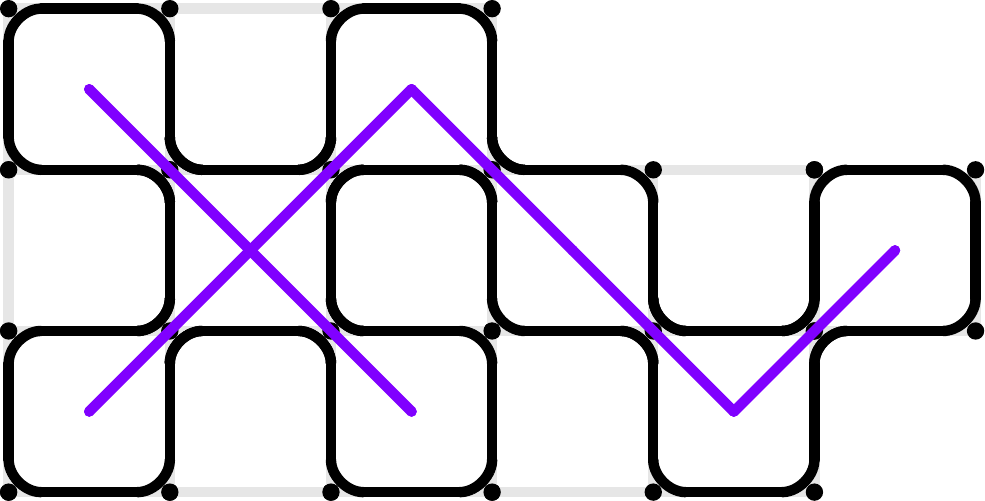}
	\caption{The solution to the example with the spanning tree of its checkering in purple.}
	\label{fig:exampleCheckeringSolved}
\end{figure}

\begin{lemma}
\label{lem:doubleTree}
The Double Turning Hamiltonian Cycle Problem in solid grid graphs admits a Hamiltonian cycle if and only if it has a valid checkering and that checkering admits a spanning tree.
\end{lemma}
\begin{proof}

First, we will prove that we can construct a Hamiltonian cycle from a spanning tree of a checkering of the graph. To do so, we will simply visit each of the vertices in an Euler tour order around the spanning tree. Each vertex in the original graph corresponds to a potential edge location in the checkering. We use this term loosely as there may not be vertices in the checkering to connect to. Around each pixel we give the vertices a clockwise ordering. From a vertex we check if that vertex corresponds to an edge in the checkering spanning tree. If not we move clockwise around our current pixel. If it is an edge, we instead consider the pixel we connect to to be our current pixel and move clockwise around that one. We know that every vertex is adjacent to exactly one or two pixels in the checkering and accordingly is visited either once or twice. This process creates a path which never crosses the spanning tree and is free to continue around the entire spanning tree, thus resulting in a single cycle as desired. An example with a checkering, spanning tree, and path can be seen in Figures~\ref{fig:exampleCheckering} and \ref{fig:exampleCheckeringSolved}.

Now we will argue that if no spanning tree exists then no Hamiltonian cycle exists. If there are two disconnected components of the checkering then this means either there are disconnected pixels in which case either the graph itself is disconnected, the graph is not checkerable, or along all connecting vertices their checkerboard edges were assigned to the odd checkering. The graph must obviously be connected and by the prior parity argument it must be cherkerable for it to admit a double turning Hamiltonian cycle. This leaves the case where we have assigned edges in our checkering graph such that it is disconnected. To do so means we would have a path through the odd checkering which separates the two parts of our checkering graph. In the same way that a Hamiltonian path cannot cross edges in the checkerboard graph, it also cannot cross edges in the odd checkering. Thus we have a vertex cut with no paths passing through it, meaning we either have more than one cycle or miss some vertices in our path.
\end{proof}

Now we merely need to show that the checkering and its spanning tree can be found in linear time.

\begin{theorem}
\label{thm:doubleTurning}
The Double Turning Hamiltonian Cycle Problem in solid grid graphs can be solved in linear time.
\end{theorem}
\begin{proof}
By Lemma~\ref{lem:doubleTree} we see that deciding if the graph is checkerable and finding a spanning tree of the checkering suffices. We assume we are given the graph embedding. First, we pick the left-most, top-most vertex in the graph and check for the 4-cycle defining the only pixel it is a part of. We construct a node in our checkerboard graph for this vertex and associate all four of the vertices around the pixel with it. Now we will perform a breath-first search over the checkered pixels in our graph. From each vertex of our current pixel, look at their exterior edges. If there are two, check whether it has been considered before and if not put that in a queue as a candidate checkerable pixel. If there is only one edge, we maintain a separate list of such degree 3 external edges. First, we check if that edge is in our list, if so we remove the edge and if not we add the edge to the list. Next, pop a candidate pixel off of the queue, verify that all four vertices are there, making it a valid pixel, and recursively check for its neighbors as before. If we ever discover a candidate pixel which is missing a vertex, then one of our properties is violated and we respond that there is no Hamiltonian cycle. Otherwise, we will have constructed a partial checkering and the bfs ordering will give us a spanning tree. Now, we need to take all edges from degree three vertices which need to be verified and check that both ends of the edge are vertices which belong to pixels in our checkering. We do the latter simply by checking whether our external degree 3 list is empty because every such edge will be added and then removed the two times it touches a valid pixel. If this is true, return that there exists a Hamiltonian cycle, and if not return false. Verifying each pixel and constructing a new node in our checkering takes constant time. Running our bfs touches each pixel, and thus each vertex in our graph a constant number of times. Each extra edge is touched twice. Thus the whole algorithm can be constructed to run in linear time.
\end{proof}

\section{Conclusion}

In this paper, we have shown that the HCPs in all of the eight semiregular tessellations are NP-complete and shown new upper and lower bounds on finding Hamiltonian paths which always turn in various grids. These generalizations we investigated lead to a large variety of open questions. Most of the restrictions from \cite{HCPtrihex} also apply to the semi-regular tessellation graphs and it would be interesting to know whether solid or super-thin versions of these graphs also admit polynomial time algorithms. We also leave open the questions of the complexity of double turning paths in square grid graphs. In addition, the dual graphs of the tessellation graphs are an obvious next target because of their regular structure and connection to discrete motion planning. One could also look at other general classes of tessellation graphs allowing more general shapes, including higher dimensional structures. We are also rather curious whether anything can be shown about finding Hamiltonian paths in aperiodic tessellation graphs.

There are also other interesting extensions of the always turning paths. The polynomial time proofs only hold for grids in the plane, however the arguments seem like they might lead to algorithms for grids on surfaces of bounded genus. It would be interesting to explore the question on square grids on a torus or other topologically distinct surfaces. In addition, the algorithm for finding double turning Hamiltonian cycles in solid square grids looks related to the number of spanning trees of certain types of graphs, as well as the potential removal of squares of edges. It would be interesting to know if it is computationally tractable to count the number of distinct double turning Hamiltonian cycles and whether it bears nice relation to other combinatorial problems. Finally, this notion of restricted turn paths can be applied to other grids or graphs with appropriate geometry.

\subparagraph*{Acknowledgments.}

We want to thank Erik Demaine for useful discussion and feedback on this research.

\bibliographystyle{alpha}
\bibliography{Tessilation}

\end{document}